\documentclass[12pt,authoryear]{article}
\usepackage[margin=2cm]{geometry}
\usepackage{authblk}

\usepackage{algorithm}
\usepackage{algpseudocode}
\usepackage{tikz}
\usetikzlibrary{calc,positioning}
\usepackage{xurl}

\usepackage{natbib}
 \bibpunct[, ]{(}{)}{,}{a}{}{,}%
\usepackage{amsmath,amssymb,amsfonts,amsthm}
\usepackage{dsfont}
\usepackage{bbm}
\usepackage{mathtools}
\usepackage{extarrows}
\usepackage{mathalpha}

\usepackage{acronym}
\usepackage{enumitem}
\usepackage{longtable}
\usepackage{tikz}
\usetikzlibrary{decorations.pathreplacing,decorations.markings}
\usepackage[normalem]{ulem}

\usepackage[dvipsnames,svgnames]{xcolor}
\colorlet{MyBlue}{DodgerBlue!80!Black}
\colorlet{MyGreen}{ForestGreen!80!Black}
\colorlet{MyRed}{OrangeRed!80!Black}

\usepackage{hyperref}
\hypersetup{
colorlinks=true,
citebordercolor=white,
linktocpage=true,
pdfstartview=FitH,
breaklinks=true,
pdfpagemode=UseNone,
pageanchor=true,
pdfpagemode=UseOutlines,
plainpages=false,
bookmarksnumbered,
bookmarksopen=false,
bookmarksopenlevel=1,
hypertexnames=true,
pdfhighlight=/O,
nesting=true,
frenchlinks, urlcolor=MyRed,linkcolor=MyGreen!,citecolor=MyBlue, 
pdfborder={0 0 0},
pdftitle={},
pdfauthor={},
pdfsubject={},
pdfkeywords={},
pdfcreator={pdfLaTeX},
pdfproducer={LaTeX with hyperref}
}

\usepackage[sort&compress,capitalize,nameinlink]{cleveref}

\crefname{enumi}{Condition}{Conditions}

\theoremstyle{plain}
\newtheorem{theorem}{Theorem}

\newtheorem{lemma}[theorem]{Lemma}
\newtheorem{proposition}[theorem]{Proposition}

\theoremstyle{definition}
\newtheorem{definition}[theorem]{Definition}

\theoremstyle{remark}
\newtheorem{remark}[theorem]{Remark}

\numberwithin{theorem}{section}

\usepackage[textwidth=30mm]{todonotes}

\DeclarePairedDelimiter{\braces}{\{}{\}}
\DeclarePairedDelimiter{\bracks}{[}{]}
\DeclarePairedDelimiter{\parens}{(}{)}

\newcommand{\debug}[1]{#1}

\newcommand{\argdot}{\,\cdot\,}

\newcommand{\diff}{\ \textup{\debug d}}
\newcommand{\ind}{\mathds{\debug 1}}
\newcommand{\eg}{e.g., }
\newcommand{\ie}{i.e., }
\newcommand{\iid}{i.i.d.\ }

\DeclareMathOperator*{\argmax}{arg\,max}

\newcommand{\naturals}{\mathbb{\debug N}}

\newcommand{\reals}{\mathbb{\debug R}}
\newcommand{\rhs}{r.h.s.\ }

\DeclareMathOperator{\Expect}{\mathsf{\debug E}}
\DeclareMathOperator{\Prob}{\mathsf{\debug P}}

\DeclareMathOperator{\Expectgen}{\mathcal{\debug E}}

\newcommand{\ahyp}{\debug a}

\newcommand{\Aset}{\debug A}
\newcommand{\auxf}{\debug \psi}

\newcommand{\bhyp}{\debug b}

\newcommand{\Bset}{\debug B}

\newcommand{\contract}{\debug \Psi}
\newcommand{\corresp}{\debug \Xi}
\newcommand{\cost}{\debug c}

\newcommand{\demand}{\debug d}
\newcommand{\Demand}{\debug D}
\newcommand{\distrdemand}{\debug F}

\newcommand{\horizon}{\debug T}

\newcommand{\ffunc}{\debug f}

\newcommand{\game}{\mathcal{\debug G}}
\newcommand{\gameenv}{\game_{\contract}}

\newcommand{\hvalue}{\debug H}
\newcommand{\hvalueret}{\hvalue^{\retailer}}
\newcommand{\hvalueretbreve}{\breve{\hvalue}^{\retailer}}
\newcommand{\hvaluemanuf}{\hvalue^{\manufacturer}}

\newcommand{\lnewsv}{\debug \ell}

\newcommand{\maximiz}{\debug \zeta}

\newcommand{\nuncens}{\debug n}
\newcommand{\uvalue}{\debug U}
\newcommand{\uvalueret}{\uvalue^{\retailer}}
\newcommand{\uvaluesmall}{\debug u}
\newcommand{\uvaluemanuf}{\uvalue^{\manufacturer}}
\newcommand{\uvaluemanufsmall}{\uvaluesmall^{\manufacturer}}

\newcommand{\oneplus}{\debug \eta}
\newcommand{\order}{\debug q}
\newcommand{\orders}{\debug Q}

\newcommand{\parameter}{\debug \theta}
\newcommand{\Parameter}{\debug \Theta}

\newcommand{\payoff}{\debug \Pi}
\newcommand{\payoffmanufacturer}{\payoff^{\manufacturer}}
\newcommand{\payoffretailer}{\payoff^{\retailer}}

\newcommand{\preddemand}{\debug G}
\newcommand{\prior}{\debug \phi}

\newcommand{\scalewholesaleprice}{\widehat{\wholesaleprice}}
\newcommand{\scaleorder}{\widehat{\order}}

\newcommand{\retailer}{\textup{\debug R}}
\newcommand{\retailprice}{\debug p}

\newcommand{\sale}{\debug s}
\newcommand{\Sale}{\debug S}

\newcommand{\sorder}{\tilde{\order}}
\newcommand{\sorders}{\tilde{\orders}}
\newcommand{\sordertwo}{ \breve{\order}}

\newcommand{\spayoff}{\debug \pi}
\newcommand{\spayoffmanufacturer}{\spayoff^{\manufacturer}}
\newcommand{\spayoffretailer}{\spayoff^{\retailer}}

\newcommand{\manufacturer}{\textup{\debug M}}

\newcommand{\survdemand}{\overline{\preddemand}}

\newcommand{\svalue}{\debug v}

\newcommand{\svalueret}{\svalue^{\retailer}}
\newcommand{\svalueretbreve}{\breve{\svalue}^\retailer}
\newcommand{\svaluemanuf}{\svalue^{\manufacturer}}
\newcommand{\swprice}{\tilde{\wholesaleprice}}

\newcommand{\transition}{\debug \Phi}
\newcommand{\ttime}{\debug t}

\newcommand{\Value}{\debug V}

\newcommand{\Valueret}{\Value^{\retailer}}
\newcommand{\Valueretbreve}{\breve{\Value}^\retailer}
\newcommand{\Valuemanuf}{\Value^{\manufacturer}}

\newcommand{\weibull}{\debug k}
\newcommand{\wholesaleprice}{\debug w}
\newcommand{\wprices}{\debug W}

\newcommand{\xvar}{\debug x}
\newcommand{\yvar}{\debug y}

\DeclareMathOperator{\Gr}{\mathsf{\debug{Gr}}}

\newcommand{\acdef}[1]{\acfi{#1}\acused{#1}}

\newacro{SPE}{Subgame perfect equilibrium}
\newacroplural{SPE}[SPE]{Subgame perfect equilibria}

\newacro{MPE}{Markov perfect equilibrium}
\newacroplural{MPE}[MPE]{Markov perfect equilibria}

\newacro{PBE}{perfect Bayesian equilibrium}
\newacroplural{PBE}[PBE]{perfect Bayesian equilibria}

\newacro{SE}{sequential equilibrium}
\newacroplural{SE}[SE]{sequential equilibria}

\newacro{NE}{Nash equilibrium}
\newacroplural{NE}[NE]{Nash equilibria}
\newacro{POS}{point-of-sale}
\newacroplural{POS}[POS]{points of sale}

\newacro{SMPS}{standardized Markov perfect solution}

\newacro{CE}{channel efficiency}

\begin{document}

\title{Dynamic Wholesale Pricing under Censored-Demand Learning} 
\author[1]{Michalis Deligiannis}
\author[2]{Marco Scarsini}
\author[3]{Xavier Venel}

\affil[1]{Department of Economics and Financial Markets, Luiss University\\
\texttt{mdeligiannis@luiss.it}
}

\affil[2]{Department of Economics and Financial Markets, Luiss University\\
\texttt{marco.scarsini@luiss.it}
}

\affil[3]{Department of AI, Data, and Decision Sciences\\ Luiss University\\
\texttt{xvenel@luiss.it}
}

\maketitle

\begin{abstract}
This paper studies dynamic wholesale pricing and ordering in a two-tier supply chain where firms share \ac{POS} data and learn about demand from censored demand data. 
When stockouts occur, unmet demand is unobserved, so the retailer’s order quantity affects not only current profits but also the informativeness of future demand signals. 
This creates a strategic interaction between pricing, ordering, and learning: the manufacturer can influence the pace of learning through wholesale prices, whereas the retailer internalizes the effect of inventory decisions on future information.
We analyze a finite-horizon dynamic game in which a manufacturer sets a wholesale price, the retailer then chooses an order quantity, demand is realized, and both firms observe sales. 
For Weibull demand with a conjugate prior, we extend a dimensionality-reduction approach from single-agent inventory learning models to a strategic supply-chain setting and use it to establish the existence of a Markov perfect equilibrium. 
For exponential demand, we further show that the equilibrium is unique and admits a recursive characterization.
Our numerical analysis shows that public learning can create conflicting incentives in the supply chain: 
In order to induce larger orders and reduce future censoring, the manufacturer chooses a wholesale price that is lower than a myopic benchmark. 
By contrast, because of its forward-looking ordering incentive, the retailer may prefer slower learning to avoid strengthening the manufacturer's future wholesale-pricing position.

\medskip
\noindent
\textbf{Keywords:}
supply chain contracting; censored demand; Bayesian learning; public learning; Markov perfect equilibrium.
\end{abstract}








%
%
%
%

\section{Introduction}

In many retail and supply-chain settings, firms observe sales rather than actual demand. 
When a retailer stocks out, lost demand is unobserved, so demand observations are censored. 
Frequent stockouts reduce the information content of observed sales data. 
Short product life cycles further limit the amount of data that can be accumulated, leading firms to rely heavily on recent sales that are only imperfectly informative about the underlying demand distribution \citep{ChuKim:OR2023}.

This paper studies a finite-horizon dynamic wholesale-pricing contract between a manufacturer and a retailer facing stochastic demand whose distribution depends on an unknown parameter. 
In each period, the manufacturer sets a wholesale price, and the retailer observes this price before choosing an order quantity.  Realized sales are publicly observed, and both parties update a common posterior over the unknown parameter from censored demand observations. 
This formulation extends the classical ``selling to the newsvendor'' framework of \citet{LarPorMSOM2001} by incorporating strategic dynamic learning from censored demand data.

The common-information structure is motivated by retailer data platforms that make realized sales and inventory information visible to upstream suppliers. Walmart describes Luminate as a first-party data analytics suite that gives merchants and suppliers a shared view of customers and products \citep{WalmartLuminate2024}. A trade-press report notes that Walmart required suppliers to move to Luminate as a replacement for a system that suppliers used to track store sales, view inventory at stores and distribution centers, and forecast store-level demand \citep{talkbusiness2023walmart}. As an illustrative example, Walmart Data Ventures reports that E.T. Browne used Scintilla data to observe declining in-stock levels and higher sales, identify a demand surge through \ac{POS} data, revise forecasts, and coordinate store-specific replenishment for stores that had run out of inventory \citep{WalmartETBrowneScintilla2025}.

The finite-horizon assumption reflects the fact that supply relationships are often revisited over time rather than fixed indefinitely.
This is consistent with evidence that contract duration in business-to-business relationships depends on economic factors and learning, and with recent evidence from the European grocery industry highlighting intensified retailer--supplier negotiations \citep{FenKri:JSR2022,McKinseyGrocery2024}.

In this strategic environment, the evolution of the common posterior belief is endogenous. 
The wholesale price affects the retailer's order quantity, and the order quantity affects the probability that the period's demand observation is censored by a stockout. 
Larger orders reduce the likelihood of censoring, thereby accelerating public learning about the demand distribution. 
Because the horizon is finite, learning has different continuation consequences depending on when it occurs: 
information revealed early in the relationship can affect many future pricing and ordering decisions, whereas information revealed near the end has a smaller impact on subsequent behavior. 
Therefore, the manufacturer can use the wholesale price not only to extract current margin, but also to influence the evolution of the common posterior belief and its own future wholesale-pricing incentives. 
Our central question is how public learning from censored demand data shapes the manufacturer's wholesale-pricing incentives and the retailer's ordering incentives, and whether the two firms value public information differently.

To our knowledge, the literature has not directly analyzed the interaction between dynamic wholesale pricing and public learning from censored demand observations. 
The Bayesian inventory literature studies learning from censored demand primarily in single-agent settings, whereas the supply-chain contracting literature typically assumes either known demand or exogenous information updates. 
Our paper brings these two strands together by studying a dynamic wholesale-pricing problem in which both firms learn from the same public history of censored demand observations.

Learning from censored demand creates substantial analytical challenges because many standard demand families do not preserve conjugacy under censoring. 
To maintain tractability, we work within the newsvendor family introduced by \citet{BraFre:MS1991}, which is closed under Bayesian updating with censored demand observations. 

We frame our problem as a dynamic Stackelberg game, and we study \acp{MPE} in which the common posterior belief is the payoff-relevant state variable. 
The analysis shows that public learning can generate conflicting incentives even when both firms observe the same sales history and share the same posterior belief.

Our main contributions are threefold.
\begin{enumerate}
\item \textbf{Dynamic contracting with censored-demand learning.}
In our finite-horizon dynamic Stackelberg game, the manufacturer sets the wholesale price before the retailer chooses the order quantity in each period. 
The manufacturer's pricing and the retailer's ordering decisions are jointly shaped by Bayesian learning from a shared history of censored demand observations. 
The model captures a public-learning environment in which neither party observes true demand directly, and current decisions affect both current payoffs and the rate of public learning through their effect on demand censoring.

\item \textbf{Equilibrium existence and properties via dimensionality reduction.}
For Weibull demand, we extend the dimensionality-reduction technique of \citet{LarPor:MS1999} to a strategic two-tier supply chain and establish the existence of an \ac{MPE} through a standardized one-parameter recursion. 
In the exponential case, the same reduction yields a unique \ac{MPE} together with a recursive characterization of equilibrium. 
A notable byproduct is that the manufacturer's equilibrium wholesale price depends on the belief state only through the number of uncensored demand observations, not through realized sales magnitudes.

\item \textbf{Economic implications of strategic learning.}
In the exponential-demand case, the recursive equilibrium characterization allows us to identify how forward-looking learning incentives distort wholesale pricing and ordering decisions. 
We compare the equilibrium with a sequence of one-period Stackelberg outcomes evaluated at the posterior updated by Bayes’ rule. 
This comparison highlights three implications. 
First, the manufacturer may lower the wholesale price relative to this myopic benchmark, to induce larger orders and reduce demand censoring. 
Second, the value of public information is asymmetric: in our numerical examples, an additional uncensored observation benefits the manufacturer, whereas its value to the retailer can be positive or negative depending on the horizon. 
Third, the retailer's forward-looking ordering need not favor learning. 
Although larger orders reduce the likelihood of stockouts and make uncensored demand observations more likely, the retailer need not value the resulting public learning positively. 
Unlike in a single-agent censored-demand newsvendor model, a more precise public belief can strengthen the manufacturer's future wholesale-pricing position, giving the retailer an incentive to slow down learning.
\end{enumerate}

The remainder of the paper is organized as follows. 
\cref{sec:literat} reviews the related literature. 
\cref{se:problem} formulates the dynamic game and introduces the solution concept. 
\cref{se:Weibull} presents the dimensionality-reduction result, establishes the existence of an \ac{MPE}, and discusses its implications. 
\cref{se:charact} characterizes the equilibrium in the exponential-demand case, \cref{sec:Numerics} presents the numerical analysis, and \cref{se:conclusions} concludes. 
Finally, \cref{se:proofs} contains the proofs.

%
%
%
%

\section{Literature Review}
\label{sec:literat}

Our work relates to research on inventory management under censored demand and to dynamic wholesale-price contracts in supply chains.

A substantial literature in inventory management studies Bayesian inventory models in which decision makers update beliefs about the parameters of demand over time.
Early contributions by \citet{Sca:AMS1959,Sca:NRLQ1960} establish state-space reductions for Bayesian inventory models under suitable structural conditions, thereby simplifying computation.
Related reductions for other observable-demand models are developed by \citet{Azo:MS1985}.
Censored demand introduces an additional difficulty: when stockouts occur, lost demand is not observed, so the decision maker learns from sales rather than from exact demand realizations.
For many standard demand families, this makes Bayesian updating analytically difficult.
To address this challenge, \citet{BraFre:MS1991} introduce the newsvendor family, which preserves conjugacy under censoring and enables tractable Bayesian inventory models in finite-horizon settings \citep[see, \eg][]{LarPor:MS1999,BisDadTok:MSOM2011,Mer:MSOM2015,BesChaMoaSS2022,ChuKim:OR2023,ZhaLiQinXuZhu:arXiv2026}.
Within this family, \citet{LarPor:MS1999} show that Weibull demand with a gamma prior admits a dimensionality-reduction result, which renders single-firm dynamic inventory problems tractable.
We build on this property and extend it from a single decision maker to a strategic two-tier supply chain.

Our paper builds on the ``selling to the newsvendor'' framework of \citet{LarPorMSOM2001}, who study wholesale-price  contracts between a manufacturer and a retailer. 
More broadly, supply chain contracting in this environment is reviewed by \citet{Cac:HORMS2003}. 

Several papers study multi-period wholesale pricing or repeated manufacturer--retailer interactions under a known demand distribution or exogenous information updates, and characterize subgame-perfect equilibria \citep[\eg][]{AnaAnuBas:MS2008,ErhKesTay:IIET2008,MarSim:POM2013,SheGraHuhNag:ORL2017}. 
Related work examines how exogenous forecast updates affect pricing and ordering decisions \citep[\eg][] {OzeUncWei:EJOR2007}. 
We refer to  \citet{SheChoMin:IJPR2019} for a broader review of supply chain contracting with information updating.
This stream of research assumes that demand information is known or evolves exogenously, rather than being learned from censored demand observations.
\citet{Han:PhD2016} studies a two-period supplier--retailer Stackelberg model with censored-demand learning under an ex-ante constant wholesale price, using a binary Bernoulli--Beta demand model. 
In contrast, we study a finite-horizon model in which wholesale prices adjust endogenously over time and we characterize the resulting Markov perfect equilibrium.

A related body of work studies dynamic contracting with information asymmetry about inventory and demand 
\citep[see, \eg][]{ZhaNagSos:OR2010,LobXia:OR2017,BenSetWan:OR2025}. 
\citet{KadOzeBen:MS2020} study a long-term vendor-managed inventory arrangement in which the manufacturer observes the retailer’s \ac{POS} and inventory data---which, as in our model, imply censored demand when stockouts occur---but the retailer also possesses additional private demand information; they propose a learn-and-screen mechanism that determines when the manufacturer should elicit this private information. 
Recent work uses online learning and bandit methods to design dynamic price contracts and study regret  \citep[see, \eg][]{CesCesOsoScaWas:IFAAMS2023,LiuRon:arXiv2024,ZhaZhuHas:MS2026}. 
These papers differ from ours because they focus either on private information or on algorithmic learning behavior, whereas our model is a dynamic game with public incomplete information in which both firms update the same posterior from the same censored demand history.

A conceptually related literature studies ratchet effects. 
Classical references include \citet{Weitzman1980,FreGuesTir:RES1985}, and \citet{LafTir:E1988}. 
In a supply-chain setting, \citet{MitShinYoon:JMR2022} show that a privately informed retailer may restrain current investment to limit future wholesale-price ratcheting.  
Our mechanism differs in that the unknown demand parameter is exogenous to both firms, the retailer's order quantity is publicly observed, and both firms update the same posterior from the same censored demand data. 
Therefore, the  resulting distortion is driven by symmetric public learning from censored demand rather than by private information or hidden actions.
Related ratchet phenomena  also arise under common uncertainty \citep{Cis:RES2018} or under symmetric learning  \citep{BhaskarRoketskiy:JET2023}. Unlike our model, however, those papers rely on hidden actions under imperfect monitoring: in \citet{Cis:RES2018}  an unobserved action affects a noisy public signal, whereas in \citet{BhaskarRoketskiy:JET2023}  effort is unobservable. 
In our model, by contrast, the retailer's order quantity is publicly observed and affects learning through demand censoring. 

Our model is a dynamic game with symmetric incomplete information and a common posterior belief that evolves publicly over time. 
More general extensive-form solution concepts for such environments include perfect Bayesian equilibrium and sequential equilibrium \citep{KreWil:E1982,FudTir:MITP1991}. 
Because in our setting the common posterior belief summarizes the payoff-relevant public history, we restrict attention to Markov strategies---that is, strategies that depend on the history only through the current common posterior belief---and we  study the corresponding \acp{MPE} \citep{MasTir:E1988-a,MasTir:JET2001}.

%
%
%
%

\section{Problem Formulation}
\label{se:problem}

We consider a dynamic wholesale-price contract between a manufacturer $\manufacturer$
and a retailer $\retailer$.
At time $0$, the two parties sign a finite-horizon  contract that specifies the following: 
In each period, the manufacturer first chooses a per-unit wholesale price. 
After observing this price, the retailer chooses an order quantity of a perishable good to meet random customer demand.
The manufacturer produces and delivers the quantity ordered by the retailer. 
The contract does not commit either party to future prices or quantities; rather, it defines the rules of interaction governing these decisions.
Inventory is not carried over across periods. 
If customer demand exceeds the order quantity in a given period,
unmet demand is not observed and not backlogged.  
The retailer shares sales data with the manufacturer, so both parties observe only realized
sales rather than the full demand realization.

More formally, we consider a finite horizon $\horizon$, and for  $\ttime\in\braces*{1,\dots,\horizon}$, the demands $\Demand_{\ttime}$ are conditionally \iid random variables given the unknown parameter $\Parameter$.
Following the Bayesian inventory literature \citep[\eg][]{LarPor:MS1999,BisDadTok:MSOM2011,Mer:MSOM2015,BesChaMoaSS2022,ChuKim:OR2023}, we assume a demand distribution of  the newsvendor family:
\begin{equation}
\label{eq:demand-distribution}	
\distrdemand(\yvar \mid \parameter) \coloneqq 
\Prob(\Demand_{\ttime} \le \yvar \mid\Parameter = \parameter) = 1 - e^{- \parameter \lnewsv(\yvar)}, \ \text{ for } \yvar \geq 0, \ \parameter > 0 , 
\end{equation}
where $\lnewsv \colon \reals_{+} \to \reals_{+}$ is a differentiable increasing function satisfying $\lim_{\yvar \searrow 0} \lnewsv(\yvar) = 0$ and $\lim_{\yvar \to \infty} \lnewsv(\yvar)$ $ = \infty$. 
An example of a newsvendor distribution is the Weibull distribution, where $\lnewsv(\yvar) = \yvar^{\weibull}$, with $\weibull > 0$. 
For $\weibull = 1$, the Weibull distribution reduces to the exponential distribution with mean $1/\parameter$.
The retailer sells at the exogenous retail price $\retailprice$, and the manufacturer incurs a per-unit production cost $\cost$, with $\retailprice>\cost>0$.

The interaction between the manufacturer and the  retailer is described by a dynamic game, which unfolds as follows: 
At time 0, the unknown parameter $\Parameter > 0$ is drawn from a prior distribution $\prior$.
In each period $\ttime = 1, \dots, \horizon$, the manufacturer sets a wholesale price $\wholesaleprice_{\ttime} \in \wprices \subset \reals_{++}$; 
the retailer observes this price and orders a quantity $\order_{\ttime} \in \orders=\reals_{+} $, which  
the manufacturer produces and delivers.
The retailer sells $ \Sale_{\ttime} = \min(\order_{\ttime}, \Demand_{\ttime}) $ units, with any unmet demand lost and unobserved. 
Both parties observe the realized sales $\sale_{\ttime}$;  hence, at  period $\ttime+1$, they have the same posterior belief about $\Parameter$, which is updated using Bayes' rule.    

We assume a conjugate prior density $\prior(\argdot)$ for the parameter $\Parameter$:
\begin{equation} 
\label{eq:conjugate}    \prior(\parameter \mid \ahyp, \bhyp) \coloneqq \frac{\bhyp^{\ahyp}\parameter^{\ahyp - 1}e^{- \bhyp\parameter}}{\Gamma(\ahyp)}\ \text{ for } \parameter > 0, \ \ahyp > 0,\  \bhyp > 0, 
\end{equation}
where $\Gamma(\argdot)$ is the Gamma function and $\ahyp$ and $\bhyp$ denote the \emph{shape} and \emph{rate} hyperparameters, respectively. 

The initial predictive  distribution function of the demand is
\begin{equation}
\label{eq:predictive-distr-func}    
\preddemand(\yvar \mid \ahyp, \bhyp) \coloneqq \Prob(\Demand_{\ttime} \le \yvar)
= 1 - \left(\frac{\bhyp}{\bhyp + \lnewsv(\yvar)} \right)^{\ahyp} \quad \text{ for } \yvar \geq 0,\ \ahyp > 0,\ \bhyp > 0, 
\end{equation}
and $\survdemand (\argdot\mid \ahyp, \bhyp) \coloneqq 1-\preddemand(\argdot\mid \ahyp,\bhyp)$ is the corresponding survival function. 
The posterior distribution of $\Parameter$ and the predictive distribution of $\Demand_{\ttime}$ have the same functional form as the prior, with different hyperparameters \citep[see][]{BraFre:MS1991}. 
If, at period $1$, the initial hyperparameters are $\ahyp_{1}$ and $\bhyp_{1}$, then, at period  $\ttime+1$, the hyperparameters become
\begin{align} 
\label{eq:a-t+1}
\ahyp_{\ttime + 1} 
&= \ahyp_{\ttime} + \ind_{ \{ { \sale_{\ttime}} < \order_{\ttime}\}}=\ahyp_{1}+\nuncens_{\ttime} \\
\label{eq:b-t+1}
\bhyp_{\ttime + 1} 
&= \bhyp_{\ttime} + \lnewsv(\sale_{\ttime})  = \bhyp_{1} + \sum_{j = 1}^{\ttime}{\lnewsv(\sale_{j})}  \quad \text{for } \ttime = 1, \ldots, \horizon-1, 
\end{align}
where $\ind_{\{\argdot \}}$ is the indicator function, $\nuncens_{\ttime}$ denotes the number of uncensored (exact-demand) observations up to period $\ttime$,  and $\sale_{\ttime}$ is the sales realization at time $\ttime$. 
The sequences $\ahyp_{\ttime}$ and $\bhyp_{\ttime}$ are weakly increasing, which implies that, for every $\ttime \ge 1$, we have $\ahyp_{\ttime} \in \Aset \coloneqq \braces*{\ahyp_{1} + \nuncens \colon \nuncens\in\naturals}$ and $\bhyp_{\ttime} \in \Bset \coloneqq [\bhyp_{1},\infty)$.

At the start of period $\ttime$, both players share a common posterior belief over $\Parameter$, summarized by the hyperparameters 
$(\ahyp_{\ttime},\bhyp_{\ttime})$. 
The manufacturer chooses $\wholesaleprice_{\ttime}$, then the retailer responds with the choice of $\order_{\ttime}$. 
The \emph{retailer's expected $\ttime$-period profit} is
\begin{equation}
\label{eq:expected-per-period-profit-retailer}
\payoffretailer(\wholesaleprice_{\ttime},\order_{\ttime} \mid \ahyp_{\ttime},\bhyp_{\ttime}) = \retailprice \, \Expect \bracks*{\min(\Demand_{\ttime}, \order_{\ttime})\mid \ahyp_{\ttime},\bhyp_{\ttime}} - \wholesaleprice_{\ttime} \order_{\ttime}.
\end{equation}
The \emph{manufacturer's $\ttime$-period profit} is
\begin{equation}
\label{eq:expected-per-period-profit-manufacturer}
\payoffmanufacturer(\wholesaleprice_{\ttime},\order_{\ttime}) = (\wholesaleprice_{\ttime}-\cost)\order_{\ttime}.
\end{equation}
We focus on Markov strategies because 
$(\ahyp_{\ttime},\bhyp_{\ttime})$ is a sufficient statistic for the conjugate posterior.  
A manufacturer's Markov strategy is a sequence of measurable functions $(\wholesaleprice_{\ttime})_{1 \le \ttime \le \horizon}$ with 
$\wholesaleprice_{\ttime} : \Aset \times \Bset \to \wprices$, whereas a retailer's 
Markov strategy is a sequence of measurable functions $(\order_{\ttime})_{1 \le \ttime \le \horizon}$ with 
$\order_{\ttime} : \Aset \times \Bset \times \wprices \to \orders$. 
The sequence of events is described in \Cref{fig:Sequ_single_p}.

\begin{figure} 
\centering
\begin{tikzpicture}[
    node distance=0.85cm,
    inner sep=2pt,
    outer sep=0pt,
    every node/.style={
      font=\sffamily\small,
      align=center
    },
    action/.style={
      draw,
      rectangle,
      rounded corners=3pt,
      minimum height=1cm,
      fill=gray!10
    },
    state/.style={
      draw,
      rectangle,
      minimum height=1cm,
      fill=white
    },
    arrow/.style={-stealth, thick},
    dashed_arrow/.style={-stealth, thick, dashed}
]

    \node (state) [state, text width=2cm]
      {State $(\ahyp_{\ttime}, \bhyp_{\ttime})$};

    \node (manufacturer) [action, right=of state, text width=2.58cm]
      {Manufacturer sets $\wholesaleprice_{\ttime}$};

    \node (retailer) [action, right=of manufacturer, text width=2.46cm]
      {Retailer sets $\order_{\ttime}$};

    \node (sales) [action, right=of retailer, text width=1.29cm]
      {Sales $\sale_{\ttime}$};

    \node (newstate) [state, right=of sales, text width=3.39cm]
      {New state $(\ahyp_{\ttime + 1}, \bhyp_{\ttime + 1})$};

    \node (augmented) [state, text width=4.7cm]
      at ($ (manufacturer)!0.5!(retailer) + (0,-1.9cm) $)
      {Augmented state $(\ahyp_{\ttime}, \bhyp_{\ttime}, \wholesaleprice_{\ttime})$};

    \draw [arrow] (state) -- (manufacturer);
    \draw [arrow] (retailer) -- (sales);
    \draw [arrow] (sales) -- (newstate);

    \draw [dashed_arrow] (manufacturer) -- (augmented);
    \draw [dashed_arrow] (augmented) -- (retailer);

\end{tikzpicture}
\vspace{-0.03cm}
  \caption{Sequential decisions and state transitions with an augmented state.}
  \label{fig:Sequ_single_p}
\end{figure}

We use the symbol $\contract$ to indicate the tuple $(\Parameter, \prior, (\distrdemand(\argdot\mid\parameter))_{\parameter}, \retailprice, \cost, \horizon, \wprices, \orders)$, and the symbol $\gameenv$ to indicate the corresponding game.

Following \citet{IskRusSch:RES2016}, we define an \ac{MPE} via dynamic programming, using the common posterior hyperparameters $(\ahyp,\bhyp)$ as the state.

Given order $\order$ and realized demand $\demand$, realized sales are $\sale=\min\{\demand,\order\}$, and the belief update is
\begin{equation}
\label{eq:Transi_funnction}
\transition(\ahyp,\bhyp,\order,\demand)
=\Bigl(\ahyp+\ind_{\{\sale<\order\}},\;\bhyp+\lnewsv(\sale)\Bigr).
\end{equation}

\begin{definition}
\label{de:belief}
Fix initial hyperparameters $(\ahyp_{1},\bhyp_{1})$ and  horizon $\horizon$.
A pair of Markov strategies $\parens*{\wholesaleprice_{\ttime}^{*},\order_{\ttime}^{*}}_{1 \le \ttime \le \horizon}$ is a \acdef{MPE} of  $\gameenv$ if there exist three sequences of value functions $\parens*{\Valuemanuf_{\ttime},\Valueretbreve_{\ttime}, \Valueret_{\ttime}}_{1 \le \ttime \le \horizon+1}$ such that, for every $\ahyp\in \Aset,\bhyp\in \Bset,\wholesaleprice \in \wprices$, the following conditions hold:
\begin{enumerate}[label={\rm (\roman*)}, ref=(\roman*)]
\item 
\label{it:cond-def-belief-i}
$\Valuemanuf_{\horizon+1}(\ahyp,\bhyp)=\Valueretbreve_{\horizon+1}(\ahyp,\bhyp,\wholesaleprice)=\Valueret_{\horizon+1}(\ahyp,\bhyp)=0$;

\item 
\label{it:cond-def-belief-ii}
for every $\ttime \in\braces*{1,\dots,\horizon}$, the manufacturer's problem is
\begin{align}
\label{eq:manufacturer-value-def} 
\begin{split}
\Valuemanuf_{\ttime}(\ahyp,\bhyp) 
&=
\max_{\wholesaleprice \in \wprices} \braces*{\payoffmanufacturer(\wholesaleprice,\order_{\ttime}^{*}(\ahyp,\bhyp,\wholesaleprice))  
+ \Expect \bracks*{ \Valuemanuf_{\ttime+1}\bigl(\transition(\ahyp, \bhyp,\order_{\ttime}^{*}(\ahyp,\bhyp,\wholesaleprice), \Demand)\bigr)\mid \ahyp,\bhyp}},  \\
\wholesaleprice_{\ttime}^{*}(\ahyp,\bhyp)
&\in
\argmax_{\wholesaleprice \in \wprices} \braces*{\payoffmanufacturer(\wholesaleprice,\order_{\ttime}^{*}(\ahyp,\bhyp,\wholesaleprice))  
+ \Expect \bracks*{ \Valuemanuf_{\ttime+1}\bigl(\transition(\ahyp, \bhyp,\order_{\ttime}^{*}(\ahyp,\bhyp,\wholesaleprice), \Demand)\bigr)\mid \ahyp,\bhyp}};
\end{split}\\
\intertext{and the retailer's problem is}
\label{eq:retailer-value-def}
\begin{split}
\Valueretbreve_{\ttime}(\ahyp,\bhyp,\wholesaleprice) 
&=
\max_{\order \in \orders} \braces*{\payoffretailer(\wholesaleprice, \order \mid \ahyp, \bhyp) 
+ \Expect \bracks*{ \Valueret_{\ttime+1}\bigl(\transition(\ahyp,\bhyp,\order,\Demand)\bigr)\mid \ahyp,\bhyp }},
\\
\order_{\ttime}^{*}(\ahyp,\bhyp,\wholesaleprice)
&\in
\argmax_{\order \in \orders}
\braces*{\payoffretailer(\wholesaleprice, \order \mid \ahyp, \bhyp) 
+ \Expect \bracks*{ \Valueret_{\ttime+1}\bigl(\transition(\ahyp,\bhyp,\order,\Demand)\bigr)\mid \ahyp,\bhyp }};
\end{split}\\
\label{eq:V-R-V-R-breve}
\Valueret_{\ttime}(\ahyp,\bhyp) 
&= \Valueretbreve_{\ttime}(\ahyp,\bhyp,\wholesaleprice_{\ttime}^{*}(\ahyp,\bhyp)).
\end{align}    
\end{enumerate}
\end{definition}

At any period $\ttime$, \cref{de:belief} involves three value functions because of the Stackelberg timing.
Because the retailer chooses an order after observing the current wholesale price, we have two value functions.
More precisely, the equilibrium is constructed by backward induction on the public belief state
$(\ahyp,\bhyp)$.  \cref{eq:manufacturer-value-def,eq:retailer-value-def}
imply that $\Valuemanuf_{\ttime}(\ahyp,\bhyp)$ (resp. $\Valueret_{\ttime}(\ahyp,\bhyp)$) is the manufacturer's (resp. retailer's) expected payoff at the start of period $\ttime$ for the strategy profile  $\parens*{\wholesaleprice_{i}^{*},\order_{i}^{*}}_{  \ttime \le i \le \horizon}$, conditionally on the belief state $(\ahyp,\bhyp)$,  whereas  $\Valueretbreve_{\ttime}(\ahyp,\bhyp,\wholesaleprice)$ is the retailer's expected payoff at period $\ttime$ after the manufacturer's  choice $\wholesaleprice$.
When $\horizon = 1$, the continuation terms are zero, and the definition reduces to the single-period Stackelberg problem studied by \citet{LarPorMSOM2001}. 
When $\horizon > 1$, the posterior hyperparameters $(\ahyp_{\ttime},\bhyp_{\ttime})$ link current decisions to future incentives through the updating equations \cref{eq:a-t+1,eq:b-t+1}. 

%
%
%
%

\section{Existence of Equilibrium and Dimensionality Reduction}
\label{se:Weibull}
Weibull demand is common in the literature. 
Building on earlier insights, \citet{LarPor:MS1999} show that a dimensionality-reduction technique can be used to analyze a single-agent newsvendor problem by representing beliefs through the shape hyperparameter~$\ahyp$ alone, thereby yielding a scalable dynamic program.
In this section, we extend this idea to our strategic setting. 
We construct a standardized newsvendor model using dimensionality reduction, and show that its solution concept can be used to obtain an \ac{MPE} of the original game. 
This approach offers two advantages:
\begin{enumerate}
[label={\rm (\roman*)}, ref=(\roman*)]
\item it simplifies the computation of the resulting \acp{MPE};
\item it enables us to establish an existence result for the equilibrium.
\end{enumerate}

To the best of our knowledge, this is the first application of dimensionality reduction to prove the existence of an equilibrium in a strategic newsvendor context. 
Moreover, solving \cref{eq:manufacturer-value-def,eq:retailer-value-def,eq:V-R-V-R-breve} via backward induction is intractable because the state space grows rapidly. 
The dimensionality reduction overcomes this challenge by leveraging the structural properties of the Weibull demand distribution.
From now on, we assume $\ahyp_{1} > \max(1,1/\weibull)$.
Let $\sorders =  \reals_{+}$.
If the conditional demand distribution $\distrdemand$ is Weibull and the prior $\prior$ is conjugate with parameter $\bhyp=1$, then the predictive demand distribution function in \cref{eq:predictive-distr-func} has the  \emph{standardized} form:
\begin{equation}
\label{eq:standardized-distribution}
\preddemand(\yvar \mid \ahyp,1) = 1 - \left( \frac{1}{1 + \yvar^{\weibull}} \right)^{\ahyp} \quad \text{for } \yvar \geq 0,\ \ahyp \in \Aset.
\end{equation}

For $\swprice\in\wprices$, $\sorder\in\sorders$, $\ahyp\in\Aset$, the  \emph{retailer’s standardized per-period expected profit} is
\begin{equation}
\label{eq:standardized-payoff-retailer}
\spayoffretailer(\swprice,\sorder\mid \ahyp)=(\ahyp - 1) \parens*{\retailprice \Expect\bracks*{\min(\Demand,\sorder) \mid \ahyp, 1}-\swprice\sorder}, 
\end{equation}
and the \emph{manufacturer’s standardized per-period profit} is
\begin{equation}
\label{eq:standardized-payoff-manufacturer}
\spayoffmanufacturer(\swprice,\sorder\mid \ahyp)=(\ahyp-1)\left(\swprice -\cost\right)\sorder.
\end{equation}

We now define the \acdef{SMPS} for  $\contract$.
A \emph{manufacturer's standardized Markov strategy}  is a sequence $\parens*{\swprice_{\ttime}}_{1 \le \ttime \le \horizon}$ with $\swprice_{\ttime} \colon \Aset\to\wprices$.
A \emph{retailer's standardized Markov strategy} is a sequence $\parens*{\sorder_{\ttime}}_{1 \le \ttime \le \horizon}$ with $\sorder_{\ttime} \colon \Aset\times\wprices\to\sorders$.  
Given $\sorder \in \sorders$, $\ahyp \in \Aset$, and $\ffunc \colon \Aset\to\reals$, we define the operator $\Expectgen$ as 
\begin{equation}
\Expectgen[\ffunc,\sorder \mid \ahyp] 
\coloneqq \frac{\ahyp-1}{\ahyp - 1/\weibull}  \preddemand(\sorder\mid \ahyp - 1/\weibull,1) \ffunc(\ahyp + 1)+ \survdemand (\sorder \mid \ahyp - 1/\weibull,1)\ffunc(\ahyp).
\end{equation}

\begin{definition}
\label{de:standardized}
Fix an initial hyperparameter $\ahyp_{1}$ and a horizon $\horizon$.
A pair of standardized Markov strategies $\parens*{\swprice_{\ttime}^{*},\sorder_{\ttime}^{*}}_{1 \le \ttime \le \horizon}$ is an  \ac{SMPS} of $\contract$ if there exist three   sequences of standardized value functions $\parens*{\svaluemanuf_{\ttime},\svalueretbreve_{\ttime},\svalueret_{\ttime}}_{1 \le \ttime \le \horizon+1}$ such that, for every $\ahyp \in \Aset,\swprice \in \wprices$, the following conditions hold: 
\begin{enumerate}[label={\rm (\roman*)}, ref=(\roman*)]
\item 
\label{it-de:standardized-1}
$\svaluemanuf_{\horizon + 1}(\ahyp) = \svalueretbreve_{\horizon + 1}(\ahyp,\swprice) = \svalueret_{\horizon + 1}(\ahyp) = 0$;
\item 
\label{it-de:standardized-2}
for every $\ttime\in\braces*{1,\dots,\horizon}$, the manufacturer's problem is
\begin{align}  
\label{eq:standardized-manufacturer}
\begin{split}
\svaluemanuf_{\ttime}(\ahyp)
&=
\max_{\swprice\in\wprices}
\braces*{\spayoffmanufacturer(\swprice,\sorder^{*}_{\ttime}(\ahyp,\swprice)\mid \ahyp)+  
\Expectgen\bracks*{ \svaluemanuf_{\ttime + 1},\sorder_{\ttime}^{*}(\ahyp,\swprice) \mid \ahyp}},\\   
\swprice_{\ttime}^{*}(\ahyp)
&\in
\argmax_{\swprice\in\wprices}
\braces*{\spayoffmanufacturer(\swprice,\sorder^{*}_{\ttime}(\ahyp,\swprice)\mid \ahyp) + \Expectgen\bracks*{ \svaluemanuf_{\ttime + 1},\sorder_{\ttime}^{*}(\ahyp,\swprice) \mid \ahyp}};
\end{split}
\intertext{and the retailer's problem is}
\label{eq:standardized-retailer}
\begin{split}
\svalueretbreve_{\ttime}(\ahyp,\swprice)
&= 
\max_{\sorder\in\sorders}
\braces*{\spayoffretailer(\swprice,\sorder \mid \ahyp) + \Expectgen\bracks*{\svalueret_{\ttime + 1}, \sorder \mid \ahyp}},\\
\sorder_{\ttime}^{*}(\ahyp,\swprice)
&\in 
\argmax_{\sorder\in\sorders}
\braces*{\spayoffretailer(\swprice,\sorder \mid \ahyp) + \Expectgen\bracks*{\svalueret_{\ttime + 1}, \sorder \mid \ahyp}};
\end{split}\\
\label{eq:v-v-breve}
    \svalueret_{\ttime}(\ahyp) 
    &= \svalueretbreve_{\ttime}(\ahyp,\swprice_{\ttime}^{*}(\ahyp)).
\end{align}
\end{enumerate}
\end{definition}

As in the single-agent problem \citep[see][]{LarPor:MS1999}, \cref{eq:standardized-manufacturer,eq:standardized-retailer,eq:v-v-breve} bear a  strong resemblance to dynamic programming, although in the standardized problem the ``transition probabilities'' do not generally sum to one; therefore, the operator $\Expectgen$ is not an expectation.

\begin{theorem}
\label{th:main}
Assume that $ \lnewsv(\yvar) = \yvar^{\weibull} $, $\wprices \coloneqq[\underline{\wholesaleprice},\overline{\wholesaleprice}]$, with $0 < \underline{\wholesaleprice} \le \cost < \retailprice \le \overline{\wholesaleprice}$.
Then:
\begin{enumerate}[label={\rm (\alph*)}, ref=(\alph*)]
\item 
\label{it:th:main-a}
There exists an \ac{SMPS} of  $\contract$. 

\item 
\label{it:th:main-b}
There exists an \ac{MPE} of $\gameenv$.

\item 
\label{it:th:main-c}
If $\parens*{\swprice_{\ttime}^{*},\sorder_{\ttime}^{*}}_{1 \le \ttime \le \horizon}$ is an \ac{SMPS} of $\contract$ with standardized value functions $\parens*{\svaluemanuf_{\ttime},\svalueretbreve_{\ttime},\svalueret_{\ttime}}_{1 \le \ttime \le \horizon+1}$,  
\begin{equation}
\label{eq:traduction}
\wholesaleprice_{\ttime}^{*}(\ahyp,\bhyp)
\coloneqq \swprice_{\ttime}^{*}(\ahyp), \quad \text{and}\quad \order_{\ttime}^{*}(\ahyp,\bhyp,\wholesaleprice)
\coloneqq \bhyp^{1/\weibull}\sorder_{\ttime}^{*}(\ahyp,\wholesaleprice),
\end{equation}
then $\parens{\wholesaleprice_{\ttime}^{*},\order_{\ttime}^{*}}_{1 \le \ttime \le \horizon}$ is an \ac{MPE} of $\gameenv$  with  value functions 
\begin{equation}
\label{eq:reduction-formulas-theorem_a}
\Valuemanuf_{\ttime}(\ahyp,\bhyp) =\frac{\bhyp^{1/\weibull}}{\ahyp-1}\svaluemanuf_{\ttime}(\ahyp),\ \Valueretbreve_{\ttime}(\ahyp,\bhyp,\wholesaleprice) =\frac{\bhyp^{1/\weibull}}{\ahyp-1}\svalueretbreve_{\ttime}(\ahyp,\wholesaleprice),  \text{ and } \Valueret_{\ttime}(\ahyp,\bhyp) =\frac{\bhyp^{1/\weibull}}{\ahyp-1}\svalueret_{\ttime}(\ahyp).
\end{equation}
\end{enumerate}
\end{theorem}

The \ac{MPE} in \cref{th:main}\ref{it:th:main-c} is obtained by solving the standardized dynamic programming equations 
\cref{eq:standardized-manufacturer,eq:standardized-retailer,eq:v-v-breve}.
These equations depend only on $\ahyp$, which evolves in unit increments. 
The continuous rate hyperparameter $\bhyp$ enters this \ac{MPE} and the value functions only through the scaling relations in \cref{eq:traduction,eq:reduction-formulas-theorem_a}.

An important implication of \cref{th:main}\ref{it:th:main-c} is that, in the \ac{MPE} constructed from an \ac{SMPS}, the manufacturer’s equilibrium wholesale price is invariant to the rate hyperparameter.
While $\bhyp_{\ttime}$ is still updated from sales magnitudes via \eqref{eq:b-t+1}, the manufacturer need not condition its price on $\bhyp_{\ttime}$.
The manufacturer’s strategy in this \ac{MPE} depends only on $\ttime$ and $\ahyp_\ttime$, which can be computed from $\ahyp_{1}$ and the number of uncensored demand observations that have occurred by period $\ttime$.
Moreover, for fixed $\ahyp$, the demand distribution is stochastically increasing in  $\bhyp$ (see \cref{eq:predictive-distr-func}), so that both the mean demand and the retailer’s equilibrium order quantity increase with $\bhyp$ (scaling with $\bhyp^{1/\weibull}$). This monotonicity also extends to the value functions. 
These results reveal an asymmetry: the manufacturer’s pricing policy is insensitive to demand scale, whereas both the retailer’s and the manufacturer’s profits increase with it through larger order quantities and a higher mean demand.

\begin{remark}
\label{re:tie-breaking}  
The difficulty in proving existence of the equilibrium is due to the fact that the game is non-zero-sum, and both the belief space and the wholesale price set are parameterized by continuous variables. 
To overcome these issues, the problem is reformulated in terms of a \acl{SMPS}; this removes the complications associated with the scale parameter. 
The existence of such an \ac{SMPS} is then established by backward induction. 
In each period $\ttime$, the retailer’s best-response correspondence may admit multiple maximizers.
Unlike the case of a  single decision maker, the maximizer selection is relevant because it affects the manufacturer's continuation value. 
To ensure well-defined continuation values for the manufacturer, a leader-favorable tie-breaking rule is imposed, under which the retailer selects the order quantity that maximizes the manufacturer’s payoff. 
This guarantees that the manufacturer’s objective attains a maximum over the feasible price set. 
This version of Stackelberg equilibrium is standard in the literature \citep[see, \eg][]{Lei:JOTA1978,BreAljHau:JOTA1988,vonZam:GEB2010,KamGen:AER2011}.
\end{remark}

%
%
%
%

\section{A Characterization of \acp{MPE}}
\label{se:charact}

In this section, we focus on the exponential case, $\lnewsv(\yvar) = \yvar$. 
We prove existence and uniqueness of the \ac{SMPS} and characterize it.
The proof of \cref{th:main}\ref{it:th:main-c}  shows that an \ac{SMPS} yields a corresponding \ac{MPE}. The theorem below further shows that this \ac{MPE} is unique.
We define
$\sordertwo_{\ttime}^{*}(\ahyp) \coloneqq \sorder_{\ttime}^{*}(\ahyp,\swprice_{\ttime}^{*}(\ahyp)),\ \text{ for all }\ttime \le \horizon, \text{ and }\ahyp\in \Aset.$

\begin{theorem}
\label{th:exponential}
Assume that $ \lnewsv(\yvar) = \yvar$, $\wprices=\reals_{++}$, and $\ahyp_1>1$. 
Then
\begin{enumerate}[label={\rm (\alph*)}, ref=(\alph*)]

\item 
\label{it:th:exponential-a}
There exists a unique \ac{SMPS} $(\swprice^{*}_{\ttime}(\ahyp), \sorder^{*}_{\ttime}(\ahyp,\swprice))$ of $\contract$, which, for all $\ttime \leq \horizon$, and all $\ahyp\in \Aset$, satisfies the following equations:
\begin{align}
\swprice_{\ttime}^{*}(\ahyp) 
&= \left(\retailprice + \svalueret_{\ttime + 1}(\ahyp+1) - \svalueret_{\ttime + 1}(\ahyp)\right) \nonumber\\ 
\label{eq:optimal_wholes_exp_c>0}
&\quad \cdot  \left[ 1 - \frac{1}{\ahyp} +\frac{\cost}{\ahyp \swprice_{\ttime}^{*}(\ahyp)}
   - \frac{\svaluemanuf_{\ttime + 1}(\ahyp + 1) - \svaluemanuf_{\ttime + 1}(\ahyp)}
   {\ahyp \left(\retailprice + \svalueret_{\ttime + 1}(\ahyp+1) - \svalueret_{\ttime + 1}(\ahyp)\right)} \right]^{\ahyp} > 0,  \\
\label{eq:z-t^{*}(a)_c>0}
\sorder^{*}_{\ttime}(\ahyp,\swprice) 
&= \max \parens*{\parens*{\frac{\retailprice + \svalueret_{\ttime + 1}(\ahyp+1) - \svalueret_{\ttime + 1}(\ahyp)}{\swprice}}^{1/\ahyp} - 1,0}, \\
\label{eq:optimal-standa-order}
\sordertwo^{*}_{\ttime}(\ahyp)
&=\parens*{\frac{\retailprice + \svalueret_{\ttime + 1}(\ahyp+1) - \svalueret_{\ttime + 1}(\ahyp)}{\swprice^{*}_{\ttime}(\ahyp)}}^{1/\ahyp} - 1 > 0.
\end{align}
The standardized value functions, $\svalueret_{\ttime}(\ahyp)$ and $ \svaluemanuf_{\ttime}(\ahyp) $ are given by  
\begin{align}
\label{eq:fs-t_exp}  \svaluemanuf_{\ttime}(\ahyp) 
&= \sum_{i = 0}^{\horizon - \ttime} \bracks*{\swprice^{*}_{\ttime + i}(\ahyp + i) - \cost\parens*{1+(\ahyp + i)\sordertwo^{*}_{\ttime + i}(\ahyp + i)}},\\
\label{eq:fr-t_exp}
\svalueret_{\ttime}(\ahyp) 
&= \sum_{i = 0}^{\horizon - \ttime} \bracks*{\retailprice - \swprice^{*}_{\ttime + i}(\ahyp + i)\parens*{1+(\ahyp + i)\sordertwo^{*}_{\ttime + i}(\ahyp + i)}},  
\end{align}
with terminal conditions $\svalueret_{\horizon+1}(\ahyp)=\svaluemanuf_{\horizon+1}(\ahyp)=0$.
\item 
\label{it:th:exponential-b}
There exists a unique \ac{MPE} $\parens*{\wholesaleprice^{*}_{\ttime}(\ahyp,\bhyp), \order^{*}_{\ttime}(\ahyp,\bhyp,\wholesaleprice)}$ of $\gameenv$, given by
\begin{equation}
\label{eq:MPE-SMPS-exponential}
\wholesaleprice_{\ttime}^{*}(\ahyp,\bhyp)
\coloneqq \swprice_{\ttime}^{*}(\ahyp) \quad \text{and}\quad \order_{\ttime}^{*}(\ahyp,\bhyp,\wholesaleprice)
\coloneqq \bhyp\sorder_{\ttime}^{*}(\ahyp,\wholesaleprice).
\end{equation}
\end{enumerate}
\end{theorem}

In the exponential case, \cref{th:exponential} provides a recursive characterization of the equilibrium. 
Notice that \cref{eq:optimal_wholes_exp_c>0} characterizes $\swprice_{\ttime}^{*}(\ahyp)$ as the unique fixed-point in $\parens{0,\retailprice + \svalueret_{\ttime + 1}(\ahyp+1) - \svalueret_{\ttime + 1}(\ahyp)}$.
Moreover, \cref{eq:z-t^{*}(a)_c>0} shows that the retailer's maximizer is unique; therefore no selection issues arise. 
Although \cref{eq:fs-t_exp} does not have the same form as the one-period manufacturer profit, it follows from substituting the manufacturer’s first-order condition into the continuation-value recursion.
A similar consideration holds for \cref{eq:fr-t_exp}.

Therefore, the standardized prices, order quantities, and standardized value functions can  be obtained by backward induction over the reachable values of $\ahyp$. 

\cref{th:exponential} also identifies the economic forces that shape equilibrium behavior. 
By \cref{eq:z-t^{*}(a)_c>0,eq:optimal-standa-order}, the retailer's ordering decision depends on the term $\retailprice+\svalueret_{\ttime+1}(\ahyp+1)-\svalueret_{\ttime+1}(\ahyp)$. 
The first component, $\retailprice$, captures the current return from selling an additional unit. The second component, $\svalueret_{\ttime+1}(\ahyp+1)-\svalueret_{\ttime+1}(\ahyp)$, captures the retailer's standardized value function of moving from $\ahyp$ to $\ahyp+1$.
Thus, the retailer's ordering decision reflects both current sales incentives and the value of improved future information.

In the single-agent setting, \citet[theorem~3(c)]{LarPor:MS1999} show that the standardized return is strictly increasing in $\ahyp$, which implies that forward-looking ordering always exceeds the myopic benchmark. 
In our strategic setting, this monotonicity need not hold for the retailer, because improved public information also affects the manufacturer's future pricing. 
As the numerical analysis in \cref{sec:Numerics} shows, the retailer's continuation-value difference can be negative, so the classical ``stock more'' result---ordering above the myopic benchmark---is not necessarily true.

The manufacturer's pricing decision reflects a related but distinct force. 
By lowering the wholesale price, the manufacturer induces a larger retailer order, which makes censoring less likely and future public beliefs less dispersed. 
Therefore, the manufacturer faces an intertemporal trade-off: raising the wholesale price increases the per-unit margin, but may reduce the retailer's order quantity, whereas lowering the wholesale price sacrifices current margin to induce larger orders, improve future information, and thereby affect future pricing opportunities. This trade-off is reflected in \cref{eq:optimal_wholes_exp_c>0}, where the equilibrium wholesale price depends on the continuation-value difference $\svaluemanuf_{\ttime+1}(\ahyp+1)-\svaluemanuf_{\ttime+1}(\ahyp)$, which captures the manufacturer's standardized gain from a more concentrated public belief.

\cref{eq:optimal-standa-order} implies that the equilibrium order quantity is always strictly positive; hence, trade remains active in every  period from $1$ to $\horizon$.
Interestingly, in the exponential case, the equilibrium wholesale price $\swprice_{\ttime}^*$ can be smaller than the production cost $\cost$ at some time $\ttime$. 
For example, when $\horizon=100$, $\ahyp_1=1.03$, $\retailprice=1$, and $\cost=0.2$, the initial equilibrium wholesale price is $\swprice_1^*(\ahyp_1)=0.1861<\cost$. 
At the initial state, the manufacturer accepts a temporary per-unit loss to induce a larger order and reduce future censoring:
This can be viewed as short-term wholesale support aimed at increasing early orders and improving demand information.
The short-run loss is offset by future profit gains generated by faster learning.
This cannot occur in the single-period model, in which there is no future learning motive.

%
%
%

\section{Numerical Results}
\label{sec:Numerics}

In this section, we compare two strategic settings.
The first is the forward-looking equilibrium described in \cref{se:charact}.
The second is a myopic benchmark in which agents update beliefs over time but, in each period, play the one-shot Stackelberg game and ignore continuation values. 
Forward-looking agents use current decisions to affect the learning process; myopic agents do not.

We focus on  the exponential-demand case ($\lnewsv(\yvar)=\yvar$), and we explore the economic implications of the unique equilibrium characterized in \cref{th:exponential}. 
Our numerical analysis focuses on three questions: 
First, how does the manufacturer use wholesale pricing to affect the informativeness of future demand observations?
Second, does better public information benefit both parties, or does the strategic interaction create asymmetries in the value of learning? 
Third, how does the retailer's ordering behavior respond to the interplay between the manufacturer's pricing and the evolving public belief?

%
%
%

\subsection{Computation and Myopic Benchmark}
\label{subsec:computation-myopic}

By \cref{th:exponential}, the unique equilibrium is computed by backward induction over the reachable values of the shape hyperparameter, whose initial value is $\ahyp_{1}$. 
For every $\ttime \le \horizon$, the set of reachable states for $
\ahyp_{\ttime}$ is $\braces*{\ahyp_{1},\ahyp_{1}+1,\dots,\ahyp_{1}+\ttime-1}$.
At each period and reachable state, we compute the equilibrium wholesale price, the retailer's best response, and the associated standardized value functions; then we iterate backward to period~$1$.

In the myopic benchmark, at each period and current public belief state, both players maximize  their current-period expected profits. 
Realized sales are still used to update beliefs, but neither player internalizes how current decisions affect future information or continuation payoffs. 
Because there is no inventory carryover or other intertemporal state variables beyond the public belief, comparing the equilibrium with this myopic benchmark isolates the effect of forward-looking learning.

Define the myopic optimal standardized order quantity as
\begin{equation}
\label{eq:myopic-standardized-order}
\sorder^{\circ}(\ahyp,\swprice)
\in
\argmax_{\sorder \in \sorders}\spayoffretailer(\swprice,\sorder \mid \ahyp),
\end{equation}
where $\spayoffretailer$ is defined as in  \cref{eq:standardized-payoff-retailer}. 
In the exponential case, this becomes
\begin{equation}
\label{eq:myopic-standardized-order-closed}
\sorder^{\circ}(\ahyp,\swprice)
=
\max\left\{\left(\frac{\retailprice}{\swprice}\right)^{1/\ahyp}-1,\,0\right\}.
\end{equation}
The corresponding myopic optimal order in the original problem is
\begin{equation}
\label{eq:myopic-order-closed}
\order^{\circ}(\ahyp,\bhyp,\wholesaleprice)
\coloneqq
\bhyp\,\sorder^{\circ}(\ahyp,\wholesaleprice),
\end{equation}
and the manufacturer's myopic optimal wholesale price is
\begin{equation}
\label{eq:myopic-wholesale}
\swprice^{\circ}(\ahyp)
\in
\argmax_{\swprice\in\wprices}
\spayoffmanufacturer\bigl(\swprice,\sorder^{\circ}(\ahyp,\swprice)\mid \ahyp\bigr),
\qquad
\wholesaleprice^{\circ}(\ahyp,\bhyp)\coloneqq \swprice^{\circ}(\ahyp),
\end{equation}
where $\spayoffmanufacturer$ is defined as in \cref{eq:standardized-payoff-manufacturer}.

%
%
%

\subsection{Learning Incentives in Wholesale Pricing}
\label{subsec:learning-driven-pricing}

We first compare the equilibrium wholesale price and its myopic benchmark. \Cref{fig:whole_path_T=5,fig:whole_path_T=10} show that, in the  nonterminal states, the wholesale price chosen by a forward-looking manufacturer is lower than the one in the myopic benchmark.  
By \cref{eq:MPE-SMPS-exponential}, the equilibrium wholesale price satisfies $\wholesaleprice_{\ttime}^{*}(\ahyp,\bhyp)=\swprice_{\ttime}^{*}(\ahyp)$, \ie it depends on the period $\ttime$ and the shape hyperparameter $\ahyp_{\ttime}$, but not on the rate hyperparameter $\bhyp_{\ttime}$. 

Under the myopic benchmark, the retailer's best response is given by \cref{eq:myopic-standardized-order-closed}. 
Given a wholesale price $\swprice \in (0,\retailprice)$, the myopic order $\sorder^{\circ}(\ahyp,\swprice)$ is strictly positive, and its sensitivity to the wholesale price becomes weaker as $\ahyp$ increases. 
Applying the comparative static result in \citet[theorem~3]{LarPorMSOM2001} to the corresponding single-period problem, we see that the manufacturer's myopic wholesale price is increasing in $\ahyp$. 
The hyperparameter  $\ahyp_{\ttime}$ increases only after uncensored observations; as a consequence, the myopic wholesale-price function  increases only after uncensored demand observations.

\begin{figure} 
\centering
\includegraphics[width=1\linewidth]{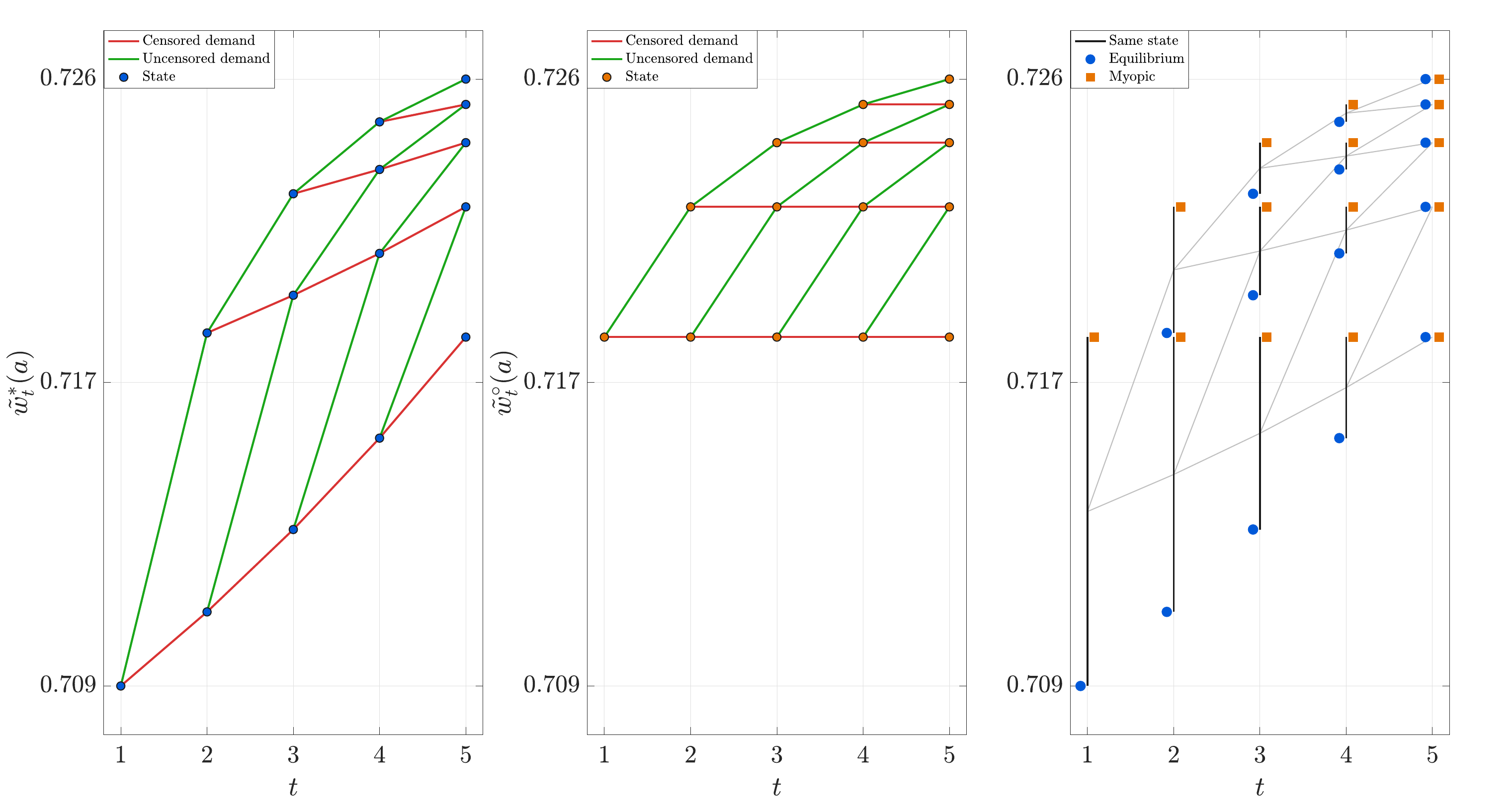}
\caption{Wholesale-price functions for $\horizon=5$, $\ahyp_{1}=2$, $\retailprice=1$, and $\cost=0.5$. Left panel: equilibrium wholesale-price function $\swprice_{\ttime}^{*}(\ahyp)$. Middle panel: myopic wholesale-price function $\swprice_{\ttime}^{\circ}(\ahyp)$. Right panel: overlay of the equilibrium and myopic wholesale prices.}
\label{fig:whole_path_T=5}
\end{figure}

\Cref{fig:whole_path_T=5} shows that in all nonterminal states, the equilibrium wholesale price lies strictly below the myopic benchmark. 
By \cref{eq:z-t^{*}(a)_c>0,eq:MPE-SMPS-exponential}, this lower wholesale price induces a larger order quantity. 
Because censoring occurs when demand exceeds the order quantity, larger orders make censoring less likely, thereby improving future learning.

This pattern is more pronounced when the horizon is longer: the equilibrium wholesale price is lower in the earlier periods and moves toward the myopic benchmark as the terminal period approaches. 

\Cref{fig:whole_path_T=10} gives the same comparison for a longer horizon and a higher production cost. 
The figure again shows that the equilibrium wholesale price is below the myopic benchmark in nonterminal states and moves toward it as the terminal period approaches.

\begin{figure}[h]
\centering
\includegraphics[width=1\linewidth]{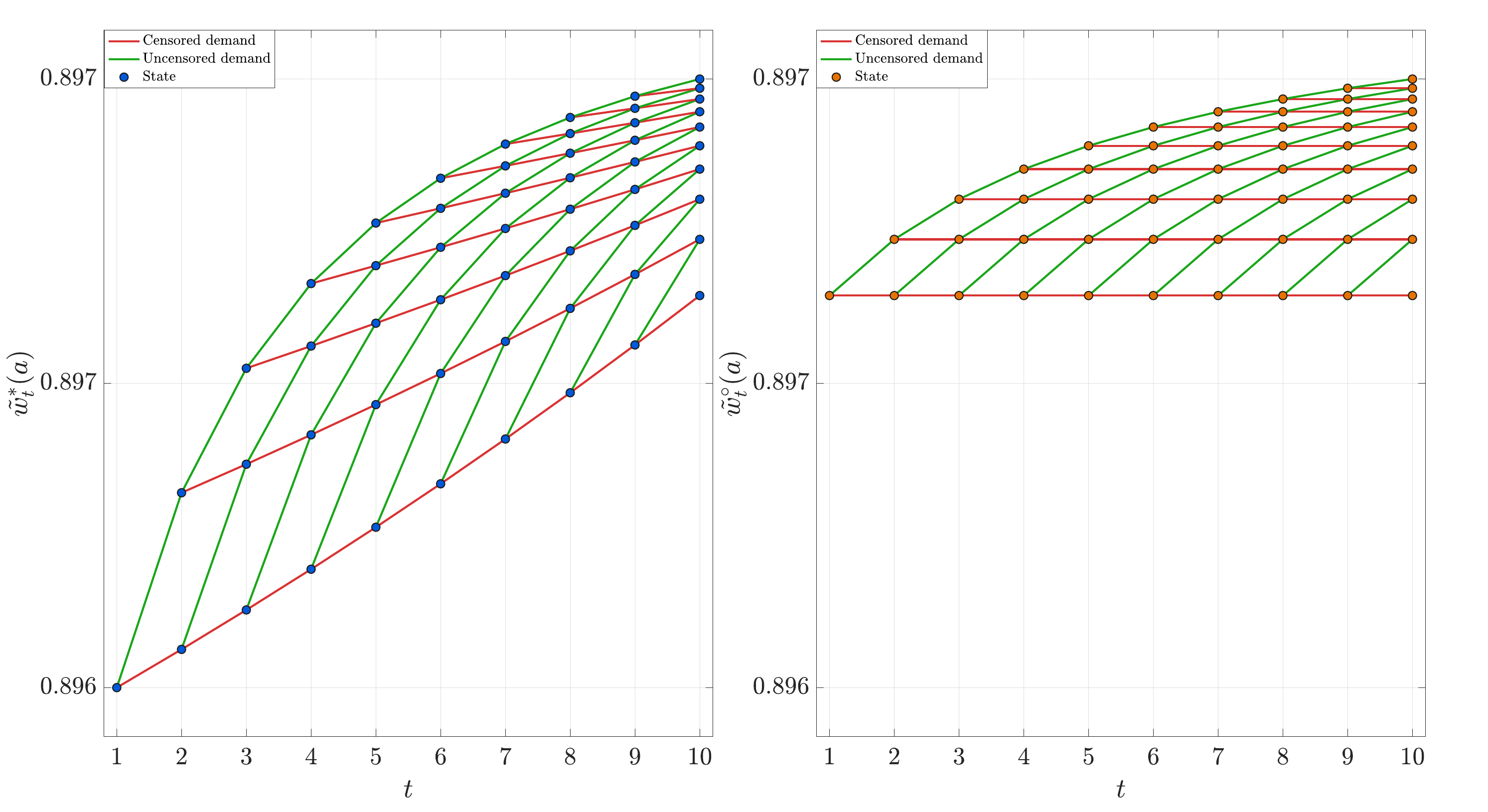}
\caption{Wholesale-price functions for $\horizon=10$, $\ahyp_{1}=2$, $\retailprice=1$, and $\cost=0.8$. 
Left panel: equilibrium wholesale-price function $\swprice_{\ttime}^{*}(\ahyp)$. Right panel: myopic wholesale-price function $\swprice_{\ttime}^{\circ}(\ahyp)$.}
\label{fig:whole_path_T=10}
\end{figure}

After a censored observation, the shape parameter $\ahyp_{\ttime}$ does not increase, whereas it does after an uncensored observation. 
Hence, a censored observation is less informative than an uncensored observation.\footnote{\citet{LarPor:MS1999} interpret larger $\ahyp$ as higher precision, using the fact that the coefficient of variation of the posterior is $1/\sqrt{\ahyp}$.}
Consistent with this, the numerical examples display lower equilibrium wholesale prices along histories with censored observations. 
Intuitively, the manufacturer's learning-oriented pricing motive is stronger when uncertainty remains high.

\Cref{fig:whole_c/p} shows how the production cost affects the equilibrium wholesale price.
Lower production costs shift the equilibrium price function downward, whereas higher production costs shift it upward. 
This pattern is consistent with a stronger learning-oriented pricing motive if inducing larger orders is less costly. 
However, \cref{subsec:ordering-responses} will show that the retailer's own strategic ordering incentive may push in the opposite direction.

\begin{figure} 
\centering
\includegraphics[width=1\linewidth]{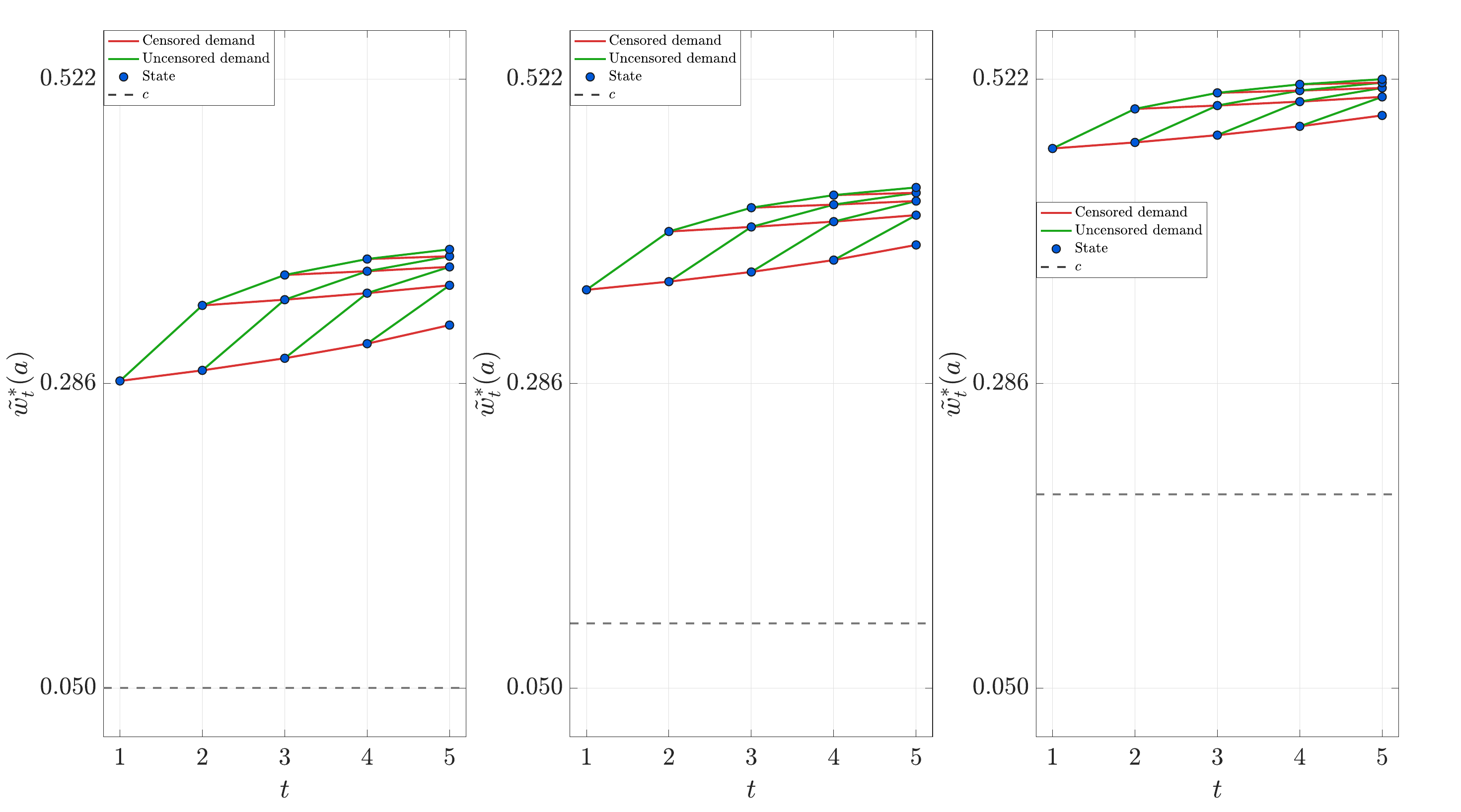}
\caption{Equilibrium wholesale-price functions for $\horizon=5$, $\ahyp_{1}=2$, $\retailprice=1$, and $\cost\in\{0.05,0.1,0.2\}$. Left panel: $\cost=0.05$. Middle panel: $\cost=0.1$. Right panel: $\cost=0.2$. The dashed line marks $\cost$.}
\label{fig:whole_c/p}
\end{figure}

%
%
%

\subsection{Value of Public Information}
\label{subsec:value-levels}

In this subsection, we study how the players' standardized value functions vary with the public belief. 
In the standardized problem, the belief is indexed by the single state variable $\ahyp$. 
\Cref{fig:stand_value_fun} shows how the standardized value functions at $\ttime=1$ vary with $\ahyp_{1}$ for different values of $\cost$. 
In all three panels, the manufacturer's standardized value function is increasing in $\ahyp_{1}$. 
The retailer's standardized value function, however, need not be monotone in $\ahyp_{1}$. 
When $\cost=0.1$, it is decreasing in $\ahyp_{1}$; when $\cost=0.4$, it is increasing in $\ahyp_{1}$; and when $\cost=0.25$, it is nonmonotone.

\begin{figure}[h] 
\centering
\includegraphics[width=1\linewidth]{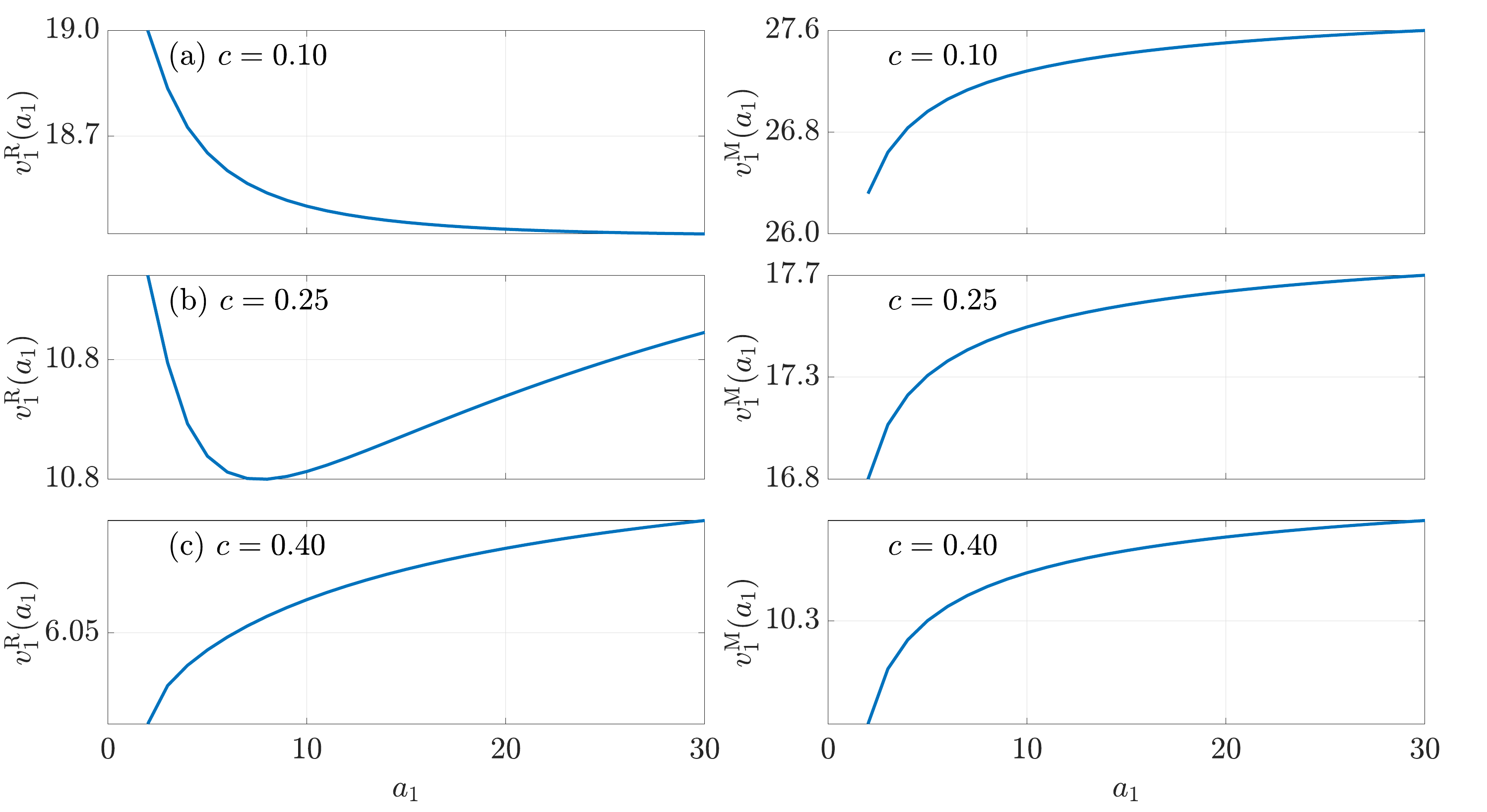}
\caption{Standardized value functions at $\ttime=1$ as functions of $\ahyp_{1}$ for $\horizon=100$, $\retailprice=1$, and $\cost\in\{0.1,0.25,0.4\}$. The left column plots the retailer's standardized value function $\svalueret_{1}(\ahyp_{1})$, and the right column plots the manufacturer's standardized value function $\svaluemanuf_{1}(\ahyp_{1})$. Panels (a), (b), and (c) correspond to $\cost=0.1$, $\cost=0.25$, and $\cost=0.4$, respectively.}
\label{fig:stand_value_fun}
\end{figure}

This pattern reflects two opposing effects for the retailer. 
On the one hand, a larger value of $\ahyp$ increases precision about the unknown demand parameter and thus tends to raise the retailer's standardized value function. 
On the other hand, when the public belief is already sufficiently concentrated, the manufacturer has a smaller incentive to keep future wholesale prices low in order to promote learning. 
This second effect is stronger when $\cost$ is low, because the manufacturer can induce learning at a lower cost. 
In the displayed examples, the second effect dominates at a low production cost, the first dominates at a high production cost, and neither dominates uniformly at an intermediate cost, yielding the nonmonotone shape in panel~(b).

%
%
%

\subsection{Standardized Marginal Value of an Additional Uncensored Observation}
\label{subsec:marginal-information}
In this subsection, we quantify the standardized marginal value of public information.
By the updating rule for $\ahyp_{\ttime}$, moving from $\ahyp_{1}$ to $\ahyp_{1}+1$ corresponds to one additional uncensored demand observation in the public belief state. 
In the standardized problem, the public belief is indexed by the single state variable $\ahyp$; therefore, we compare the standardized value functions at $\ahyp_{1}$ and $\ahyp_{1}+1$.
We interpret
 $\svalueret_{2}(\ahyp_{1}+1)-\svalueret_{2}(\ahyp_{1})
\quad\text{and}\quad
\svaluemanuf_{2}(\ahyp_{1}+1)-\svaluemanuf_{2}(\ahyp_{1})$ as the standardized marginal value of moving from $\ahyp_{1}$ to $\ahyp_{1}+1$ for the retailer and the manufacturer, respectively.

\begin{figure}[h] 
\centering
\includegraphics[width=1\linewidth]{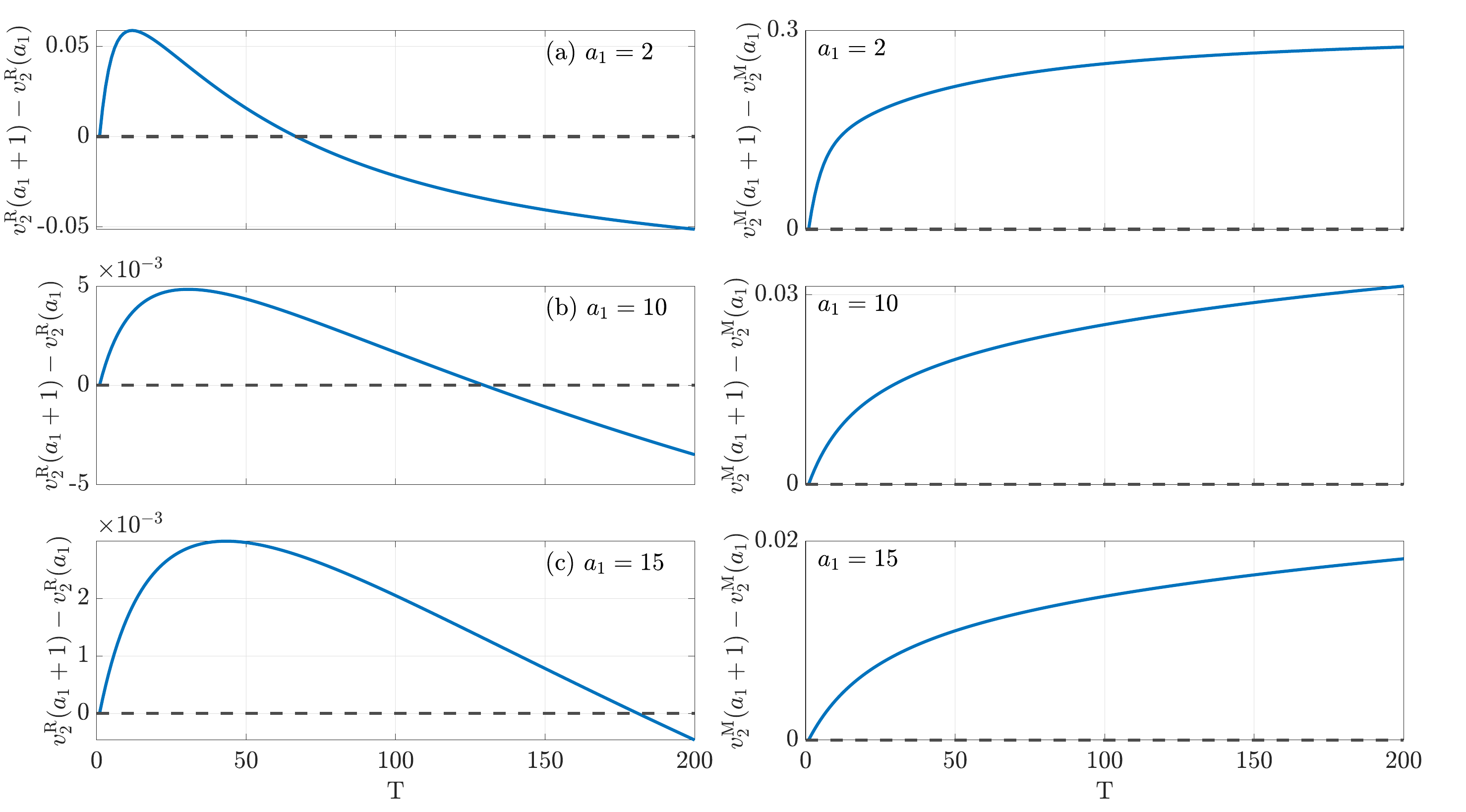}
\caption{Differences in standardized value functions as functions of the horizon $\horizon$ for $\retailprice=1$, $\cost=0.25$, and $\ahyp_{1}\in\{2,10,15\}$. The left column plots $\svalueret_{2}(\ahyp_{1}+1)-\svalueret_{2}(\ahyp_{1})$, and the right column plots $\svaluemanuf_{2}(\ahyp_{1}+1)-\svaluemanuf_{2}(\ahyp_{1})$. Panels (a), (b), and (c) correspond to $\ahyp_{1}=2$, $\ahyp_{1}=10$, and $\ahyp_{1}=15$, respectively.}
\label{fig:diffe_stand_value_fun}
\end{figure}

\Cref{fig:diffe_stand_value_fun} illustrates a clear asymmetry between the two players.
For the manufacturer, the standardized marginal value $\svaluemanuf_{2}(\ahyp_{1}+1)-\svaluemanuf_{2}(\ahyp_{1})$
is positive and increases with $\horizon$. 
Thus, in these examples, an additional uncensored observation at the beginning of the relationship is valuable to the manufacturer, and this difference is larger when the horizon is longer.
For the retailer, the difference
$\svalueret_{2}(\ahyp_{1}+1)-\svalueret_{2}(\ahyp_{1})$ 
is positive for short horizons but decreases with $\horizon$ and eventually becomes negative for all three displayed values of  $\ahyp_{1}$. 
Therefore, an additional uncensored observation may benefit the retailer when the relationship is short, but reduces the retailer’s standardized value function when the relationship is long. 
This effect is strongest when $\ahyp_{1}$ is small and becomes weaker as $\ahyp_{1}$ increases, which is consistent with the fact that the incremental informational gain from moving from $\ahyp_{1}$ to $\ahyp_{1}+1$ is smaller when the public belief is already concentrated.

The horizon pattern in \Cref{fig:diffe_stand_value_fun} reflects the trade-off discussed in the previous subsection. 
An additional uncensored observation improves demand information, but it also strengthens the manufacturer's future pricing position. 
For short horizons, the first effect dominates, so the retailer's standardized marginal value of information is positive. 
For long horizons, the second effect can dominate, making the retailer's marginal value of information negative. By contrast, in the displayed examples, the manufacturer's standardized marginal value of information is positive throughout. 
The next subsection shows how this asymmetry is reflected in the retailer's ordering incentives.

%
%
%

\subsection{Ordering Responses under Dynamic Learning}
\label{subsec:ordering-responses}

A central result in the Bayesian inventory literature is that a forward-looking retailer orders more than a myopic retailer. 
In the exponential-demand setting studied by \citet[theorem~3(c),(e)]{LarPor:MS1999}, this result is a direct consequence of the fact that the standardized return is strictly increasing in the shape hyperparameter. 
Our strategic setting shows that this conclusion need not survive when learning also affects future wholesale pricing, that is, for a given wholesale price, the retailer may prefer to order less than in the myopic benchmark.

The retailer's decision is affected not only by the manufacturer's choice of the wholesale price, but also by the effect that the ordered quantity has on public learning. 
To understand how dynamic learning affects the retailer's order, it is useful to compare the equilibrium order with the myopic benchmark and separate the effect of the manufacturer's lower wholesale price from the retailer's own strategic ordering incentives. Accordingly, we decompose the gap between the equilibrium order and the myopic order as follows:
\begin{equation}
\label{eq:decomposition-order-gap}
\underbrace{\sordertwo_{\ttime}^{*}(\ahyp)-\sorder^{\circ}(\ahyp,\swprice^{\circ}(\ahyp))}_{\text{total order gap}}
=
\underbrace{\Bigl(\sordertwo_{\ttime}^{*}(\ahyp)-\sorder^{\circ}(\ahyp,\swprice_{\ttime}^{*}(\ahyp))\Bigr)}_{\text{forward-looking effect}}
+
\underbrace{\Bigl(\sorder^{\circ}(\ahyp,\swprice_{\ttime}^{*}(\ahyp))-\sorder^{\circ}(\ahyp,\swprice^{\circ}(\ahyp))\Bigr)}_{\text{price effect}}.
\end{equation}

The first term on the right-hand side is the forward-looking effect: it compares the equilibrium order with the myopic order at the equilibrium wholesale price $\swprice_{\ttime}^{*}(\ahyp)$. 
The second term is the price effect: it captures the change in the myopic order generated by the gap between the equilibrium and myopic wholesale prices.

\begin{figure}[h] 
\centering
\includegraphics[width=1\linewidth]{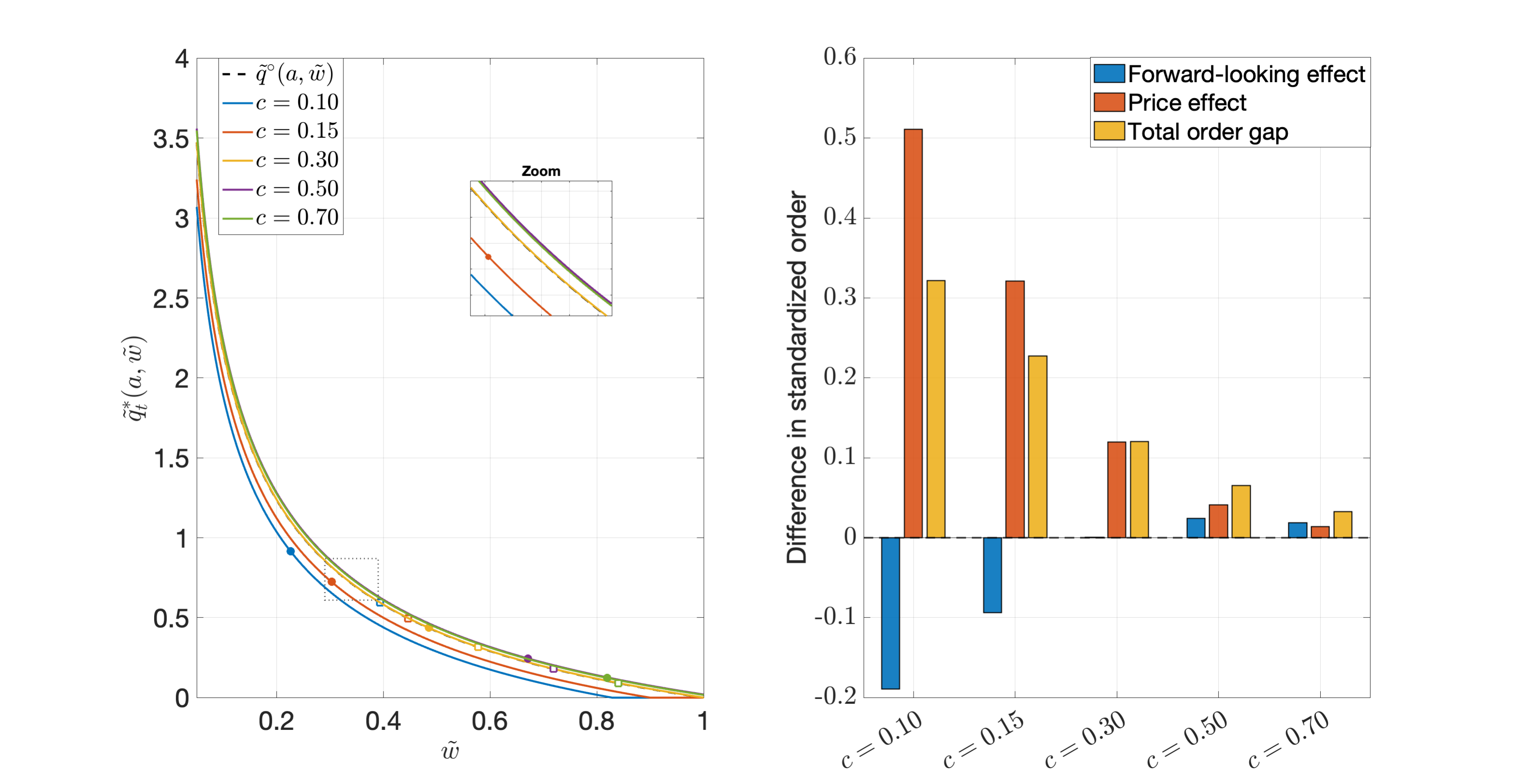}
\caption{Retailer's standardized response and decomposition of the equilibrium--myopic order gap for $\horizon=100$, $\ahyp_{1}=2$, $\ttime=1$, $\retailprice=1$, and $\cost\in\{0.1,0.15,0.3,0.5,0.7\}$. 
Left panel: retailer's standardized best-response functions $\sorder^{*}_{1}(\ahyp_{1},\swprice)$ and the standardized myopic retailer response $\sorder^{\circ}(\ahyp_{1},\swprice)$. 
Right panel: decomposition of $\sordertwo^{*}_{1}(\ahyp_{1})-\sorder^{\circ}(\ahyp_{1},\swprice^{\circ}(\ahyp_{1}))$ into the forward-looking effect, the price effect, and the total order gap.}
\label{fig:rachet_effe}
\end{figure}
\Cref{fig:rachet_effe} shows that the retailer's ordering response under dynamic learning can go in either direction. 
The left panel fixes $\ttime=1$ and $\ahyp_{1}=2$ and compares the retailer's standardized best-response function with the standardized myopic response for different values of $\cost$. 
For low production cost ($\cost=0.1,0.15$), the equilibrium best-response function lies below the myopic benchmark over the relevant range of wholesale prices, so, for a fixed wholesale price, the retailer would order less than in the myopic benchmark. 
When $\cost=0.3$, the two responses are nearly indistinguishable, whereas for $\cost=0.5$ and $\cost=0.7$ the equilibrium best response lies slightly above the myopic benchmark. 
Thus, in the displayed examples, the retailer's forward-looking effect can be negative, negligible, or positive, depending on the production cost.

The right panel quantifies the decomposition in \cref{eq:decomposition-order-gap}. 
The forward-looking effect is negative for $\cost=0.1$ and $\cost=0.15$, essentially zero for $\cost=0.3$, and positive for $\cost=0.5$ and $\cost=0.7$. 
By contrast, the price effect is positive in all cases and is strongest when production costs are low. 
This pattern is consistent with \cref{subsec:learning-driven-pricing}: 
a lower production cost widens the equilibrium--myopic price gap, which amplifies the price effect.

The decomposition also clarifies why the forward-looking effect differs from the single-agent problem. 
At a fixed wholesale price, choosing a larger order makes an uncensored observation more likely, which implies faster learning. 
In the single-agent problem, this is beneficial because the retailer's continuation value increases with the shape parameter. 
In our setting, however, \Cref{fig:diffe_stand_value_fun} shows that the retailer's standardized marginal value of an additional uncensored observation can be negative. 
For this reason, the retailer may prefer to induce slower learning, even though the manufacturer prefers the opposite.
The decomposition in \cref{eq:decomposition-order-gap}, illustrated in the right panel of \Cref{fig:rachet_effe}, shows that the total order gap is positive even when the forward-looking retailer has an  incentive to reduce orders.
The manufacturer's equilibrium price reduction more than offsets the retailer's incentive to slow down learning.

\begin{remark}
\label{re:ratchet}
The forward-looking effect identified in this subsection is related to the ratchet effect in dynamic incentive problems. 
In the classical principal-agent literature on the ratchet effect, the agent's current performance reveals information that can later be exploited by the principal, so the agent strategically restrains current performance \citep{Weitzman1980,FreGuesTir:RES1985,LafTir:E1988}.

In a supply-chain setting, \citet{MitShinYoon:JMR2022} show that a privately informed retailer may curtail current demand-enhancing investment in order to conceal favorable market conditions and avoid a higher future wholesale price.
This intuition is also related to \citet{Han:PhD2016}, who notes that a retailer may prefer to ``create censorship'' to avoid an unfavorable future wholesale price. 
In \citet{Cis:RES2018}, a long-run player and the market share a common prior and learn from a noisy public signal, but the player's actions affect the evolution of the signal under imperfect monitoring. 
In \citet{BhaskarRoketskiy:JET2023}, the worker and the firm learn symmetrically about job difficulty, but effort is unobservable.

Our model is different in that the demand parameter $\Parameter$ is unknown to both players, the retailer's order quantity is publicly observed, and both parties update the same posterior from the same censored demand data. Thus, the relevant state variable is a common public belief at every date. There is no private information, no hidden action, and no privately controlled signal.

Nevertheless, a ratchet-like force may arise endogenously.  Holding the wholesale price fixed, a larger order quantity makes an uncensored observation more likely and  speeds up public learning. 
When faster public learning can lead to higher future wholesale prices, the retailer may prefer to order less today in order to slow down that learning. 
In this sense, the forward-looking effect in \cref{eq:decomposition-order-gap} is analogous to a ratchet effect: current behavior is distorted because it affects future wholesale prices through the evolution of beliefs. 
The difference is that, in our model, this distortion is driven entirely by symmetric public learning from censored demand rather than by asymmetric information.
\end{remark}

%
%
%

\section{Concluding Remarks}
\label{se:conclusions}
The Bayesian inventory literature on censored-demand learning has mostly studied a single decision maker. This paper studies a finite-horizon dynamic wholesale-pricing contract in which a manufacturer and a retailer observe sales, update a common posterior from censored demand observations, and interact strategically over time. For Weibull demand, we establish existence of an \ac{MPE} using dimensionality reduction. 
For exponential demand, we prove uniqueness of this \ac{MPE} and recursively characterize equilibrium wholesale prices, order quantities, and value functions.

Our analysis shows that public learning may create conflicting incentives in the supply chain. 
In our model, the manufacturer uses the wholesale price not only to extract current margin, but also to influence how much future sales observations reveal about demand. 
This generates an asymmetry in the value of public information. 
In our numerical analysis, more precise information benefits the manufacturer, whereas its value to the retailer can be positive, negative, or can vary with the horizon. 
As a result, the classical stock-more result from Bayesian inventory management need not survive in a strategic setting: the retailer may prefer slower learning when faster public learning strengthens the manufacturer’s future pricing position.

The model also suggests several directions for future research. 
One is to introduce retailer-side private information, which would allow a closer comparison between our mechanism and the classical ratchet literature. 
Other extensions include richer contract forms, such as buyback or revenue-sharing contracts, and inventory carryover across periods.

%
%
%
%

\section{Proofs}
\label{se:proofs}

\subsection{Berge's Theorems}
\label{suse:Berge}

To prove  
\cref{th:main}, we will use three key results, two of them from \citep[chapter VI, section 3]{Ber:DUNOD1959}. 
For the reader's convenience, we briefly restate these results following the style of \citet{AliBor:SPRINGER2006}.

Given a correspondence $\maximiz$, the symbol $\Gr(\maximiz)$ denotes its graph. 
\begin{lemma}[\protect{\citet[lemma~17.30]{AliBor:SPRINGER2006}}]
\label{le:up-up-up}
Let $\maximiz \colon \wprices \rightrightarrows \sorders $ be an upper hemicontinuous correspondence between topological spaces with nonempty compact values, and let
\begin{equation*}
\uvalue \colon \Gr(\maximiz) \to \reals
\end{equation*}
be an upper semicontinuous objective function. 
Define the function $\svalue \colon \wprices \to \reals $ by
\[
\svalue(\wholesaleprice) = \max_{\sorder \in \maximiz(\wholesaleprice)} \uvalue(\wholesaleprice,\sorder).
\]
Then the function $\svalue$ is upper semicontinuous.
\end{lemma}

\begin{theorem}[\protect{\citet[theorem~17.31]{AliBor:SPRINGER2006}}]
\label{thm:cont-cont-cont-up}
Let $\corresp \colon \wprices \rightrightarrows \sorders $ be a continuous correspondence between topological spaces with nonempty compact values, and let $\uvalue \colon  \Gr(\corresp) \to \reals$
be a continuous objective function.
Define the \emph{value function} $\svalue \colon \wprices \to \reals $ by
\[
\svalue(\wholesaleprice) = \max_{\sorder \in \corresp(\wholesaleprice)} \uvalue(\wholesaleprice,\sorder),
\]
and the correspondence of maximizers $\maximiz \colon \wprices \rightrightarrows \sorders$ by
\[
\maximiz(\wholesaleprice) = \braces*{\sorder \in \corresp(\wholesaleprice) \colon \uvalue(\wholesaleprice,\sorder) = \svalue(\wholesaleprice)}.
\]

Then:
\begin{enumerate}
\item The value function $\svalue$ is continuous.
  
\item The ``$\argmax$'' correspondence $\maximiz$ has nonempty compact values.
  
\item If either $\uvalue$ has a continuous extension to all of $\wprices\times\sorders $, or if $\sorders$ is Hausdorff, then the ``$\argmax$'' correspondence $\maximiz$ is upper hemicontinuous.
  \end{enumerate}
\end{theorem}

A correspondence $\corresp \colon \wprices \rightrightarrows \sorders$ is called \emph{weakly measurable} if, for every open set $Z\subseteq \sorders$,
\[
\{\wholesaleprice\in\wprices :
\corresp(\wholesaleprice)\cap Z\neq\varnothing\}
\]
is a Borel subset of $\wprices$.

A function $\uvalue\colon \wprices\times\sorders\to\reals$ is \emph{Carathéodory} if, for every $\sorder\in\sorders$, the map
$\wholesaleprice\mapsto \uvalue(\wholesaleprice,\sorder)$ is measurable and, for every $\wholesaleprice\in\wprices$, the map
$\sorder\mapsto \uvalue(\wholesaleprice,\sorder)$ is continuous.
\begin{theorem}[\protect{\citet[theorem~18.19]{AliBor:SPRINGER2006}}]
\label{thm:measurable-maximum}
Let $\wprices$ and $\sorders$ be separable metrizable spaces endowed with their respective Borel $\sigma$-algebras.
Let $\corresp \colon \wprices \rightrightarrows \sorders$ be a weakly measurable correspondence with nonempty compact values, and let
$\uvalue \colon \wprices \times \sorders \to \reals$
be a Carathéodory function.
Define the \emph{value function}
$\svalue \colon \wprices \to \reals$ by
\[
\svalue(\wholesaleprice)
=
\max_{\sorder\in\corresp(\wholesaleprice)}
\uvalue(\wholesaleprice,\sorder),
\]
and the correspondence of maximizers
$\maximiz \colon \wprices \rightrightarrows \sorders$ by
\[
\maximiz(\wholesaleprice)
=
\braces*{
\sorder\in\corresp(\wholesaleprice)
:
\uvalue(\wholesaleprice,\sorder)
=
\svalue(\wholesaleprice)
}.
\]

Then:
\begin{enumerate}
\item The value function $\svalue$ is measurable.

\item The ``$\argmax$'' correspondence $\maximiz$ has nonempty compact values.

\item The ``$\argmax$'' correspondence $\maximiz$ is measurable and admits a measurable selector.
\end{enumerate}
\end{theorem}

%
%
%

\subsection{Dimensionality Reduction}
\label{suse:proof-dimensionality-reduction}

We use Berge’s maximum theorem to obtain existence and upper hemicontinuity of the retailer’s best-response correspondence. 

\begin{proof}[Proof of \cref{th:main}] 
\ref{it:th:main-a} 
The standardized value functions are set to be $0$ at period $\horizon+1$. 

We assume that the standardized value functions at period $\ttime+1$ are well defined, and we construct $\parens*{\swprice^{*}_{\ttime}(\ahyp), \sorder^{*}_{\ttime}(\ahyp,\swprice), \svaluemanuf_{\ttime}(\ahyp),\svalueretbreve_{\ttime}(\ahyp,\swprice), \svalueret_{\ttime}(\ahyp)}$ 
so that \cref{it-de:standardized-2} in \cref{de:standardized} is satisfied at period~$\ttime$.

We first show that the retailer's action space can be restricted to a compact set that depends on $\ahyp$ and $\ttime$. 
Define
\begin{equation}
\label{eq:M-r_exist}  \uvalueret_{\ttime}(\swprice,\sorder \mid \ahyp)\coloneqq \spayoffretailer(\swprice,\sorder \mid \ahyp) + \frac{\ahyp-1}{\ahyp - 1/\weibull}  \preddemand(\sorder \mid \ahyp - 1/\weibull,1) \svalueret_{\ttime + 1}(\ahyp + 1)+ \survdemand (\sorder \mid \ahyp - 1/\weibull,1)\svalueret_{\ttime + 1}(\ahyp).
\end{equation}
Considering  \cref{eq:standardized-payoff-retailer} and replacing $\min(\Demand,\sorder)$ with $\Demand$ and $\swprice$ with $\underline{\wholesaleprice}$, we get
\begin{equation}
\label{eq:ineq-Pi-s}
\spayoffretailer(\swprice,\sorder \mid \ahyp) \le   
(\ahyp-1) \parens*{\retailprice \int_{0}^{\infty} \yvar \diff\preddemand(\yvar \mid \ahyp,1) -\underline{\wholesaleprice}\sorder}.
\end{equation}
From \cref{eq:M-r_exist,eq:ineq-Pi-s}, we obtain
\begin{equation}  \label{eq:exist_boun-M-r-t}
\uvalueret_{\ttime}(\swprice,\sorder \mid \ahyp) \leq (\ahyp-1)\parens*{\retailprice \int_{0}^{\infty} \yvar \diff\preddemand(\yvar \mid \ahyp,1)-\underline{\wholesaleprice}\sorder} + \frac{\ahyp-1}{\ahyp - 1/\weibull}  \svalueret_{\ttime + 1}(\ahyp + 1)+ \svalueret_{\ttime + 1}(\ahyp).
\end{equation}
Because $\swprice \ge \underline{\wholesaleprice}$, the \rhs of  \cref{eq:exist_boun-M-r-t} is bounded above by a linear function of $\sorder$ with strictly negative slope, uniformly over  $\swprice \in  \wprices$.
Hence, there exists $\overline{\sorder}_{\ttime}(\ahyp)$ such that any action strictly greater than   $\overline{\sorder}_{\ttime}(\ahyp)$ leads to a strictly negative payoff for every $\swprice \in  \wprices$.
Choosing $\sorder = 0 $ yields a payoff   $\svalueret_{\ttime + 1}(\ahyp)\geq 0$. 
Hence, any action strictly greater than $\overline{\sorder}_{\ttime}(\ahyp)$ is strictly dominated. 
This implies that the retailer's action space can be restricted to the compact set $\Xi_{\ttime}(\ahyp) \coloneqq [0, \overline{\sorder}_{\ttime}(\ahyp)] $.
 
To apply Berge's  maximum theorem to our problem, we fix $\ahyp$ and $\ttime$, and consider
\begin{itemize}
\item 
the constant correspondence $\swprice \mapsto  \Xi_{\ttime}(\ahyp)$, which is  continuous;

\item 
the objective function $(\swprice,\sorder) \mapsto \uvalueret_{\ttime}(\swprice,\sorder \mid \ahyp)$, which is continuous, because $\preddemand$ and $\spayoffretailer$ are  continuous.
\end{itemize}

Applying \cref{thm:cont-cont-cont-up}, we see that 
\begin{equation}
\label{eq:v-breve-def}
\svalueretbreve_{\ttime}(\ahyp,\swprice) \coloneqq \max_{\sorder} \uvalueret_{\ttime}(\swprice,\sorder \mid \ahyp)
\end{equation}
is continuous in $\swprice$ and the argmax correspondence, denoted by $\maximiz_{\ttime}(\ahyp,\swprice)$, has nonempty compact values and is upper hemicontinuous in $\swprice$. Moreover, the objective function is  Carath\'eodory and the constant correspondence  $\swprice \mapsto  \Xi_{\ttime}(\ahyp)$ is weakly Borel measurable. 
By \cref{thm:measurable-maximum},
the correspondence $\swprice \mapsto  \maximiz_{\ttime}(\ahyp,\swprice)$ of maximizers in  \cref{eq:v-breve-def} is measurable.
To select an element of $\maximiz_{\ttime}(\ahyp,\swprice)$, we define
\begin{equation}\label{eq:sup_Payof-t-c>0-toto}
\uvaluemanufsmall_{\ttime}(\swprice,\sorder \mid \ahyp)\coloneqq(\ahyp-1) (\swprice-\cost) \sorder + \frac{\ahyp-1}{\ahyp - 1/\weibull} \svaluemanuf_{\ttime + 1}(\ahyp + 1) \preddemand(\sorder \mid \ahyp - 1/\weibull,1) + \svaluemanuf_{\ttime + 1}(\ahyp) \survdemand (\sorder \mid \ahyp - 1/\weibull,1).
\end{equation}

Consider now
\begin{itemize}
\item the correspondence $\swprice \mapsto \maximiz_{\ttime}(\ahyp,\swprice)$, which is measurable and upper hemicontinuous and has nonempty compact values; 
\item the objective function $(\swprice,\sorder) \mapsto 
\uvaluemanufsmall_{\ttime}(\swprice,\sorder \mid \ahyp)$, which is continuous in $(\swprice,
\sorder)$, because $\preddemand$ is continuous.  
\end{itemize}

By \cref{le:up-up-up}, we have that 
\begin{equation}
\label{eq:selection}
\uvaluemanuf_{\ttime}(\swprice \mid \ahyp) \coloneqq \max_{\sorder \in \maximiz_{\ttime}(\ahyp,\swprice)} \uvaluemanufsmall_{\ttime}(\swprice,\sorder \mid \ahyp)  
\end{equation}
is upper semicontinuous in $\swprice$.
Moreover, the objective function is  Carath\'eodory and the correspondence $\maximiz_{\ttime}(\ahyp,\cdot)$ is weakly Borel measurable. Hence, by \cref{thm:measurable-maximum},
the set $\maximiz_{\ttime}^{*}(\ahyp,\swprice)$ of maximizers in  \cref{eq:selection} is measurable and there exists  a measurable selection $\sorder_{\ttime}^{*}(\ahyp,\swprice)$ of $\maximiz_{\ttime}^{*}(\ahyp,\swprice)$.
Therefore,  $\sorder_{\ttime}^{*}(\ahyp,\swprice)$ satisfies
$\uvaluemanuf_{\ttime}(\swprice \mid \ahyp)= \uvaluemanufsmall_{\ttime}(\swprice,\sorder_{\ttime}^{*}(\ahyp,\swprice)\mid \ahyp)$.

We now look at the manufacturer's optimization problem.
Because $\wprices \coloneqq[\underline{\wholesaleprice},\overline{\wholesaleprice}]$ is compact and $\swprice \mapsto \uvaluemanuf_{\ttime}(\swprice \mid \ahyp)$ is upper semicontinuous, one can define $\svaluemanuf_{\ttime}(\ahyp) \coloneqq \max_{\swprice} \uvaluemanuf_{\ttime}(\swprice \mid \ahyp)$ and pick a maximizer $\swprice_{\ttime}^{*}(\ahyp) \in \argmax_{\swprice} \uvaluemanuf_{\ttime}(\swprice \mid \ahyp)$.
We can then define $\svalueret_{\ttime}(\ahyp) \coloneqq \svalueretbreve_{\ttime}(\ahyp,\swprice_{\ttime}^{*}(\ahyp))$. 
This proves the existence of the value function and the corresponding maximizing strategies at period $\ttime$. 
This concludes the induction proof.

\noindent
\ref{it:th:main-b} follows from the construction in  \ref{it:th:main-c}.

\noindent
\ref{it:th:main-c} 
By \ref{it:th:main-a}, an \ac{SMPS} of $\contract$ exists. 
Let $\parens*{\swprice_{\ttime}^{*},\sorder_{\ttime}^{*}}_{1 \le \ttime \le \horizon}$ be such an \ac{SMPS} with standardized value functions 
$\parens*{\svaluemanuf_{\ttime},\svalueretbreve_{\ttime},\svalueret_{\ttime}}_{1 \le \ttime \le \horizon+1}$.
Define, for all $\ttime\in\braces*{1,\dots,\horizon}$, $\ahyp\in\Aset$, $\bhyp\in\Bset$, $\wholesaleprice\in\wprices$, 
\begin{equation}
\label{eq:ut-zt*}
\scalewholesaleprice_{\ttime}(\ahyp,\bhyp)
=\swprice_{\ttime}^{*}(\ahyp) \quad\text{and}\quad
\scaleorder_{\ttime}(\ahyp,\bhyp,\wholesaleprice)
=\bhyp^{1/\weibull}\sorder_{\ttime}^{*}(\ahyp,
\wholesaleprice),
\end{equation}

and for all $\ttime\in\braces*{1,\dots,\horizon+1}$, $\ahyp\in\Aset$, $\bhyp\in\Bset$, $\wholesaleprice\in\wprices$, 
\begin{equation}
\label{reduction-formulas}
\hvalueret_{\ttime}(\ahyp,\bhyp) = \frac{\bhyp^{1/\weibull}}{\ahyp-1}\svalueret_{\ttime}(\ahyp),\
\hvaluemanuf_{\ttime}(\ahyp,\bhyp)= \frac{\bhyp^{1/\weibull}}{\ahyp-1}\svaluemanuf_{\ttime}(\ahyp), \quad\text{and}\quad
\hvalueretbreve_{\ttime}(\ahyp,\bhyp,\wholesaleprice) = \frac{\bhyp^{1/\weibull}}{\ahyp-1}\svalueretbreve_{\ttime}(\ahyp,\wholesaleprice).
\end{equation}

To prove that $\parens{\scalewholesaleprice_{\ttime},\scaleorder_{\ttime}}_{1 \le \ttime \le \horizon}$  is an \ac{MPE}
with the above value functions, we check that the conditions of \cref{de:belief} are satisfied as a direct consequence of their counterpart in \cref{de:standardized}.
First, we have $\hvalueret_{\horizon+1}(\ahyp,\bhyp)=\hvaluemanuf_{\horizon+1}(\ahyp,\bhyp)=\hvalueretbreve_{\horizon+1}(\ahyp,\bhyp,\wholesaleprice)=0$; hence \cref{it:cond-def-belief-i} in \cref{de:belief} is satisfied. 
Next, we check that \cref{it:cond-def-belief-ii} in \cref{de:belief} is verified at any period $\ttime \in \{1, \dots, \horizon\}$.
From \cref{eq:retailer-value-def}, we get
\begin{align}
\begin{split}
\label{Proof_reduction-retailer_B}
&\quad
\max_{\order}\bigg\{ \payoffretailer(\wholesaleprice,\order \mid \ahyp,\bhyp) +  \int_{0}^{\order} \hvalueret_{\ttime + 1} \left(\ahyp+1,\bhyp+\yvar^{\weibull} \right)\diff \preddemand(\yvar \mid \ahyp,\bhyp)  \\
&\qquad\qquad+
\hvalueret_{\ttime + 1} \left( \ahyp,\bhyp+\order^{\weibull} \right)  \survdemand (\order\mid \ahyp,\bhyp)\bigg\}   
\end{split}
\\
\begin{split}
&\overset{\textup{(b)}}{=} 
\max_{\order}\bigg\{ \ \retailprice \int_{0}^{\order} \survdemand (\yvar \mid \ahyp,\bhyp)\diff \yvar -\wholesaleprice \order + \frac{\svalueret_{\ttime + 1} \left(\ahyp+1 \right)}{\ahyp} \int_{0}^{\order} \left(\bhyp + \yvar^{\weibull} \right)^{1/\weibull} \diff \preddemand(\yvar \mid \ahyp,\bhyp) \\
&\qquad\qquad+
\frac{\svalueret_{\ttime + 1} (\ahyp)}{\ahyp-1} \left(\bhyp + \order^{\weibull} \right)^{1/\weibull} \survdemand (\order\mid \ahyp,\bhyp) \bigg\}. 
\end{split}
\\
\begin{split}
&\overset{\textup{(c)}}{=} 
\bhyp^{1/\weibull} \max_{\sorder} \bigg\{  \retailprice \int_{0}^{\sorder} \survdemand (\xvar \mid \ahyp, 1)\diff\xvar  -  \wholesaleprice \sorder + \frac{\svalueret_{\ttime + 1}(\ahyp + 1)}{\ahyp} \int_{0}^{\sorder} (1 + \xvar^{\weibull})^{1/\weibull} \diff\preddemand(\xvar\mid \ahyp,1)  \\
&\qquad\qquad\qquad+ 
\frac{\svalueret_{\ttime + 1}(\ahyp)}{\ahyp-1}(1 + \sorder^{\weibull})^{1/\weibull} \survdemand ( \sorder \mid \ahyp,1) \bigg\}
\end{split}
\\
\begin{split}
&\overset{\textup{(d)}}{=} 
\bhyp^{1/\weibull} \max_{\sorder} \bigg\{\retailprice\int_{0}^{\sorder} {\survdemand (\xvar \mid \ahyp,1)} \diff \xvar -\wholesaleprice\sorder  + \frac{\svalueret_{\ttime + 1}(\ahyp + 1) }{\ahyp - 1/\weibull}  \preddemand(\sorder \mid \ahyp - 1/\weibull, 1) \\
&\qquad\qquad\qquad+ \frac{\svalueret_{\ttime + 1}(\ahyp)}{\ahyp-1} \survdemand (\sorder \mid \ahyp - 1/\weibull, 1) \bigg\}
\end{split}
\\
\label{Proof_reduction-retailer_R}
&=
\frac{\bhyp^{1/\weibull}}{\ahyp-1} \max_{\sorder} \bigg\{\spayoffretailer(\wholesaleprice,\sorder \mid \ahyp) + \frac{\ahyp-1}{\ahyp - 1/\weibull} \svalueret_{\ttime + 1}(\ahyp + 1) \preddemand(\sorder \mid \ahyp - 1/\weibull, 1) + \svalueret_{\ttime + 1}(\ahyp) \survdemand (\sorder \mid \ahyp - 1/\weibull, 1) \bigg\}\\
&= 
\frac{\bhyp^{1/\weibull}}{\ahyp-1} \svalueretbreve_{\ttime}(\ahyp,\wholesaleprice)=\hvalueretbreve_{\ttime}(\ahyp,\bhyp,\wholesaleprice).
\end{align} 
The equality (b) follows from \cref{reduction-formulas}. The equality (c) follows from the changes of variable $\sorder = \order /\bhyp^{1/\weibull}$ and $\xvar  = \yvar/ \bhyp^{1/\weibull}$, which imply
\begin{align}
\label{eq:nonral_Distri_{1}}
\survdemand (\order \mid \ahyp, \bhyp) = 
\survdemand (\sorder \mid \ahyp,1) \text{ and } \int_{0}^{\order} \left( \bhyp + \yvar^{\weibull} \right)^{1/\weibull} \diff\preddemand(\yvar\mid \ahyp, \bhyp)= 
\bhyp^{1/\weibull}\int_{0}^{\sorder}(1 + \xvar^{\weibull})^{1/\weibull} \diff\preddemand(\xvar \mid \ahyp,1).
\end{align}
The equality (d) follows from $(1 + \sorder^{\weibull})^{1/\weibull} \survdemand (\sorder \mid \ahyp,1)= \survdemand (\sorder \mid \ahyp - 1/\weibull, 1)$ and
\begin{equation}
\label{eq:nonral_Distri_3}
\int_{0}^{\sorder}(1 + \xvar^{\weibull})^{1/\weibull} \diff\preddemand(\xvar \mid \ahyp,1) 
= 
\frac{\ahyp}{\ahyp - 1/\weibull} \preddemand(\sorder \mid \ahyp - 1/\weibull, 1).
\end{equation}
Finally, we recognize in \cref{Proof_reduction-retailer_R} the \rhs of \cref{eq:standardized-retailer}. 
Moreover, $\sorder^{*}_{\ttime}(\ahyp,\wholesaleprice)$ satisfies \cref{it-de:standardized-2} in \cref{de:standardized}, which implies that it is a maximizer in \cref{Proof_reduction-retailer_R}.
As a consequence, $\scaleorder_{\ttime}(\ahyp,\bhyp,\wholesaleprice)$ is a maximizer of \cref{Proof_reduction-retailer_B}.

Next, consider the manufacturer’s problem in \cref{eq:manufacturer-value-def}:
\begin{align}
\label{Proof_reduction-manufacturer_B}
\begin{split}
&\quad 
\max_{\wholesaleprice} \bigg\{ \payoffmanufacturer(\wholesaleprice,\scaleorder_{\ttime}(\ahyp,\bhyp,\wholesaleprice)) + \int_{0}^{\scaleorder_{\ttime}(\ahyp,\bhyp,\wholesaleprice)} \hvaluemanuf_{\ttime + 1}(\ahyp+1, \bhyp+\yvar^{\weibull})\, \diff\preddemand(\yvar \mid \ahyp, \bhyp)  \\
&\qquad\qquad+ \hvaluemanuf_{\ttime + 1}(\ahyp, \bhyp+(\scaleorder_{\ttime}(\ahyp,\bhyp,\wholesaleprice))^{\weibull})\, \survdemand (\scaleorder_{\ttime}(\ahyp,\bhyp,\wholesaleprice) \mid \ahyp, \bhyp) \bigg\}
\end{split}
\\
\begin{split}
&\overset{\textup{(e)}}{=} 
\max_{\wholesaleprice}\bigg\{(\wholesaleprice-\cost)\bhyp^{1/\weibull}\sorder^{*}_{\ttime}(\ahyp,\wholesaleprice) 
+ \int_{0}^{\bhyp^{1/\weibull}\sorder^{*}_{\ttime}(\ahyp,\wholesaleprice)}{\svaluemanuf_{\ttime + 1}(\ahyp + 1)\frac{(\bhyp + \yvar^{\weibull})^{1/\weibull}}{\ahyp}} \diff\preddemand(\yvar \mid \ahyp, \bhyp)  \\
&\qquad\qquad+ 
\svaluemanuf_{\ttime + 1}(\ahyp)\frac{(\bhyp + \bhyp(\sorder^{*}_{\ttime}(\ahyp,\wholesaleprice))^{\weibull})^{1/\weibull}}{\ahyp-1} \survdemand (\bhyp^{1/\weibull}\sorder^{*}_{\ttime}(\ahyp,\wholesaleprice)\mid \ahyp, \bhyp) \bigg\}
\end{split}
\\
\begin{split}
&= 
\max_{\wholesaleprice}\bigg\{(\wholesaleprice-\cost)\bhyp^{1/\weibull}\sorder^{*}_{\ttime}(\ahyp,\wholesaleprice) 
+ \frac{\svaluemanuf_{\ttime + 1}(\ahyp + 1)}{\ahyp} \int_{0}^{\bhyp^{1/\weibull}\sorder^{*}_{\ttime}(\ahyp,\wholesaleprice)}{(\bhyp + \yvar^{\weibull})^{1/\weibull}} \diff\preddemand(\yvar \mid \ahyp, \bhyp)  \\
&\qquad\qquad+ 
\svaluemanuf_{\ttime + 1}(\ahyp)\frac{\bhyp^{1/\weibull}}{\ahyp-1}(1 + (\sorder^{*}_{\ttime}(\ahyp,\wholesaleprice))^{\weibull})^{1/\weibull} \survdemand (\bhyp^{1/\weibull}\sorder^{*}_{\ttime}(\ahyp,\wholesaleprice)\mid \ahyp, \bhyp) \bigg\}
\end{split}
\\
\begin{split}
&\overset{\textup{(f)}}{=} 
\max_{\wholesaleprice}\bigg\{(\wholesaleprice-\cost)\bhyp^{1/\weibull}\sorder^{*}_{\ttime}(\ahyp,\wholesaleprice) 
+ \frac{\svaluemanuf_{\ttime + 1}(\ahyp + 1)}{\ahyp}  \frac{\ahyp }{\ahyp - 1/\weibull}\bhyp^{1/\weibull}\preddemand(\sorder^{*}_{\ttime}(\ahyp,\wholesaleprice) \mid \ahyp - 1/\weibull, 1)  \\
&\qquad\qquad+ 
\svaluemanuf_{\ttime + 1} (\ahyp) \frac{\bhyp^{1/\weibull}}{\ahyp-1}\survdemand (\sorder^{*}_{\ttime}(\ahyp,\wholesaleprice) \mid \ahyp - 1/\weibull, 1) \bigg\}
\end{split}
\\
 \begin{split}
  &= 
 \frac{\bhyp^{1/\weibull}}{\ahyp-1} \max_{\wholesaleprice}\bigg\{(\wholesaleprice-\cost)(\ahyp-1)\sorder^{*}_{\ttime}(\ahyp,\wholesaleprice) 
 + \svaluemanuf_{\ttime + 1}(\ahyp + 1)\frac{\ahyp-1}{\ahyp - 1/\weibull}\preddemand(\sorder^{*}_{\ttime}(\ahyp,\wholesaleprice) \mid \ahyp - 1/\weibull, 1)  \\
 &\qquad\qquad+ 
 \svaluemanuf_{\ttime + 1} (\ahyp)\survdemand (\sorder^{*}_{\ttime}(\ahyp,\wholesaleprice) \mid \ahyp - 1/\weibull, 1) \bigg\}
 \end{split}
 \\
&\overset{\textup{(g)}}{=} 
\frac{\bhyp^{1/\weibull}}{\ahyp-1} \svaluemanuf_{\ttime}(\ahyp)= \hvaluemanuf_{\ttime}(\ahyp,\bhyp).
\end{align}
Equality~(e) stems from \cref{eq:ut-zt*,reduction-formulas}. Equality~(f) can be obtained by adapting  \cref{eq:nonral_Distri_{1},eq:nonral_Distri_3}, and Equality~(g) follows from dynamic programming for $\svaluemanuf_{\ttime}(\ahyp)$. 
Finally, we see that $\scalewholesaleprice_{\ttime}(\ahyp,\bhyp)$ is a maximizer in  \cref{Proof_reduction-manufacturer_B} because $\swprice_{\ttime}^{*}(\ahyp)$ satisfies \cref{eq:standardized-manufacturer} in \cref{de:standardized}. 
Therefore, 
\begin{equation*}
\hvalueretbreve_{\ttime}(\ahyp,\bhyp,\scalewholesaleprice_{\ttime}(\ahyp,\bhyp))
=\hvalueretbreve_{\ttime}(\ahyp,\bhyp,\swprice^{*}_{\ttime}(\ahyp))
=\frac{\bhyp^{1/\weibull}}{\ahyp-1}\svalueretbreve _{\ttime}(\ahyp,\swprice_{\ttime}^{*}(\ahyp))
=\frac{\bhyp^{1/\weibull}}{\ahyp-1}\svalueret_{\ttime}(\ahyp)
=\hvalueret_{\ttime}(\ahyp,\bhyp).
\end{equation*}
\end{proof}

\begin{proposition}
\label{pr:unique-SMPS}  
Let $\lnewsv(\yvar) = \yvar^{\weibull}$. 
If the \ac{SMPS} $\parens*{\swprice_{\ttime}^{*}, \sorder_{\ttime}^{*}}_{1 \le \ttime \le \horizon}$ of $\contract$ is unique, then there exists a unique \ac{MPE} $\parens*{\wholesaleprice_{\ttime}, \order_{\ttime}}_{1 \le \ttime \le \horizon}$ of $\gameenv$. 
This \ac{MPE} is  defined by \eqref{eq:traduction} and its value functions satisfy \cref{eq:reduction-formulas-theorem_a}.
\end{proposition}

\begin{proof}
Assume that the \ac{SMPS} $\parens*{\swprice_{\ttime}^{*},\sorder_{\ttime}^{*}}_{1 \le \ttime \le \horizon}$ of $\contract$ with standardized value functions $\parens*{\svalueret_{\ttime},\svalueretbreve_{\ttime},\svaluemanuf_{\ttime}}_{1 \le \ttime \le \horizon+1}$ is unique. 
Then, by \cref{th:main}, the pair  $(\scalewholesaleprice_{\ttime},\scaleorder_{\ttime})$ defined in \cref{eq:traduction} is an \ac{MPE} of $\gameenv$ with value functions  $\parens*{\hvalueret_{\ttime},\hvalueretbreve_{\ttime},\hvaluemanuf_{\ttime}}_{1 \le \ttime \le \horizon+1}$, (following the notations used in the proof of \cref{th:main}).
Let $\parens*{\wholesaleprice_{\ttime}, \order_{\ttime}}_{1 \le \ttime \le \horizon}$ be any \ac{MPE} of $\gameenv$ with value functions  $\parens*{\Valueret_{\ttime},\Valueretbreve_{\ttime},\Valuemanuf_{\ttime}}_{1 \le \ttime \le \horizon+1}$. 
We show, by backward induction, that the two \acp{MPE}s and their value functions coincide. 
At period $\horizon+1$, the value functions coincide because they are all equal to zero.
We now assume that, for some 
$\ttime \in \braces*{1, \dots, \horizon}$ the value functions are equal for $\ttime + 1, \dots,\horizon+1$, and we prove  that $(\wholesaleprice_{\ttime}, \order_{\ttime}) = (\scalewholesaleprice_{\ttime},\scaleorder_{\ttime})$, which implies that the value functions also coincide at $\ttime$. 
Replacing $\hvalueret_{\ttime + 1}(\argdot,\argdot)$ with  $\Valueret_{\ttime + 1}(\argdot,\argdot)$ in \cref{Proof_reduction-retailer_R}, we obtain
\begin{equation}
\begin{split}
&\quad \max_{\order}\bigg\{ \payoffretailer(\wholesaleprice,\order \mid \ahyp,\bhyp) +  \int_{0}^{\order} \Valueret_{\ttime + 1} (\ahyp+1,\bhyp+\yvar^{\weibull}) \diff\preddemand(\yvar \mid \ahyp, \bhyp) +\Valueret_{\ttime + 1} ( \ahyp,\bhyp+\order^{\weibull})  \survdemand (\order \mid \ahyp, \bhyp) \bigg\}\\
&= 
\frac{\bhyp^{1/\weibull}}{\ahyp-1} \max_{\sorder} \bigg\{\spayoffretailer(\wholesaleprice,\sorder \mid \ahyp) + \frac{\ahyp-1}{\ahyp - 1/\weibull} \svalueret_{\ttime + 1}(\ahyp + 1) \preddemand(\sorder \mid \ahyp - 1/\weibull, 1) + \svalueret_{\ttime + 1}(\ahyp) \survdemand (\sorder \mid \ahyp - 1/\weibull, 1) \bigg\}.
\end{split}
\end{equation} 
We know that $\order_{\ttime}(\ahyp,\bhyp,\wholesaleprice)$ is a maximizer on the left-hand side of the equation.
Therefore, $\order_{\ttime}(\ahyp,\bhyp,\wholesaleprice)/\bhyp^{1/\weibull}$ is a maximizer on the right-hand side of the equation. 
Uniqueness of the \ac{SMPS} implies that 
\begin{equation}
\frac{\order_{\ttime}(\ahyp,\bhyp,\wholesaleprice)}{\bhyp^{1/\weibull}}=\sorder^{*}_{\ttime}(\ahyp,\wholesaleprice)\ \iff \order_{\ttime}(\ahyp,\bhyp,\wholesaleprice)
=\bhyp^{1/\weibull}\sorder^{*}_{\ttime}(\ahyp,\wholesaleprice)
=\scaleorder_{\ttime}(\ahyp,\bhyp,\wholesaleprice).
\end{equation}
Moreover, 
\begin{equation}
\Valueretbreve_{\ttime}(\ahyp,\bhyp,\wholesaleprice)
= 
\frac{\bhyp^{1/\weibull}}{\ahyp-1} \svalueretbreve_{\ttime}(\ahyp,\wholesaleprice)
=
\hvalueretbreve_{\ttime}(\ahyp,\bhyp,\wholesaleprice).
\end{equation}
We now consider the problem of the manufacturer at period $\ttime$. 
Replacing $\hvaluemanuf_{\ttime + 1}(\argdot,\argdot)$ by $\Valuemanuf_{\ttime + 1}(\argdot,\argdot)$ in \cref{Proof_reduction-manufacturer_B}, we obtain
\begin{equation}
\begin{split}
&  \max_{\wholesaleprice} \bigg\{ \payoffmanufacturer(\wholesaleprice,\order_{\ttime}(\ahyp,\bhyp,\wholesaleprice)) 
+ \int_0^{\order_{\ttime}(\ahyp,\bhyp,\wholesaleprice)} \Valuemanuf_{\ttime + 1}(\ahyp+1, \bhyp+\yvar^{\weibull})\, \diff\preddemand(\yvar \mid \ahyp, \bhyp) \\
&\qquad + \Valuemanuf_{\ttime + 1}(\ahyp,\bhyp{+}(\order_{\ttime}(\ahyp,\bhyp,\wholesaleprice))^{\weibull})\, \survdemand (\order_{\ttime}(\ahyp,\bhyp,\wholesaleprice) \mid \ahyp, \bhyp) \bigg\} 
\\
&=  
\frac{\bhyp^{1/\weibull}}{\ahyp-1} \max_{\wholesaleprice} \bigg\{\spayoffmanufacturer(\wholesaleprice,\sorder^{*}_{\ttime}(\ahyp,\wholesaleprice)\mid \ahyp)  + \frac{\ahyp-1}{\ahyp - 1/\weibull} \svaluemanuf_{\ttime + 1}(\ahyp + 1)\preddemand(\sorder^{*}_{\ttime}(\ahyp,\wholesaleprice) \mid \ahyp - 1/\weibull, 1)\\
& \hspace{1cm}+ \svaluemanuf_{\ttime + 1}(\ahyp) \survdemand (\sorder^{*}_{\ttime}(\ahyp,\wholesaleprice) \mid \ahyp - 1/\weibull, 1) \bigg\}. 
\end{split}
\end{equation}
Any maximizer on the left-hand side of the equation is a maximizer on the right-hand side. 
By uniqueness of $\swprice^{*}_{\ttime}(\ahyp)$, it follows that, for every $\bhyp\in \Bset$, 
$\wholesaleprice_{\ttime}(\ahyp,\bhyp)=\swprice^{*}_{\ttime}(\ahyp)
=\scalewholesaleprice_{\ttime}(\ahyp,\bhyp)$.
Moreover, because the value function satisfies dynamic programming, we deduce the equality of the value functions at period $\ttime$:
\begin{align}
\Valuemanuf_{\ttime}(\ahyp,\bhyp) 
&= \frac{\bhyp^{1/\weibull}}{\ahyp-1}\svaluemanuf_{\ttime}(\ahyp)
= \hvaluemanuf_{\ttime}(\ahyp,\bhyp) 
, \text{ and } 
\Valueret_{\ttime}(\ahyp,\bhyp) 
= \frac{\bhyp^{1/\weibull}}{\ahyp-1} \svalueret_{\ttime}(\ahyp)
=
\hvalueret_{\ttime}(\ahyp,\bhyp).
\end{align}
This concludes the proof.
\end{proof}

\subsection{Characterization of \ac{MPE} with  Exponential Demand}
\label{suse:Proof-Exponential}

We first use backward induction to establish  existence and uniqueness of the equilibrium; then we derive the characterization of \cref{th:exponential}.

\begin{proposition}
\label{prop:Norma-properties-exp-c>0}
Assume $ \lnewsv(\yvar) = \yvar $, $\ahyp \in \Aset$, $\wprices=\reals_{++}$, and $\sorders=\reals_{+}$. 
Then there exists a unique \ac{SMPS} $ (\swprice^{*}_{\ttime}(\ahyp), \sorder^{*}_{\ttime}(\ahyp,\swprice))_{\ttime\in\{1,\dots,\horizon\}}$ and  standardized value functions $\svalueret_{\ttime}(\ahyp),\svalueretbreve_{\ttime}(\ahyp,\swprice),\svaluemanuf_{\ttime}(\ahyp)$ that satisfy the following conditions  for all $ \ttime \in \{1, \dots, \horizon\}$  and $\ahyp \in \Aset$:

\begin{enumerate}[label={\rm (\roman*)}, ref=(\roman*)]
\item 
\label{it:w*-bou-c>0}
$\swprice_{\ttime}^{*}(\ahyp)\in \parens*{0,\retailprice+ \svalueret_{\ttime + 1}(\ahyp+1)-\svalueret_{\ttime + 1}(\ahyp)}$,

\item 
\label{it:z*-bou-c>0}
$0<\sordertwo_{\ttime}^{*}(\ahyp)<\infty$,

\item
\label{it:z>0-c>0}
$\retailprice+ \svalueret_{\ttime}(\ahyp+1)-\svalueret_{\ttime}(\ahyp)>0$, 

\item
\label{it:w<wmax-c>0}
$\retailprice-\cost+ \svalueret_{\ttime}(\ahyp+1)+\svaluemanuf_{\ttime}(\ahyp+1)>\svalueret_{\ttime}(\ahyp)+\svaluemanuf_{\ttime}(\ahyp)$.
\end{enumerate}
\end{proposition}

\begin{proof}[Proof of \cref{prop:Norma-properties-exp-c>0}]
We use backward induction to construct the sequence of functions $ (\swprice^{*}_{\ttime}(\ahyp),$ $\sorder^{*}_{\ttime}(\ahyp,\swprice),\svalueret_{\ttime}(\ahyp),\svalueretbreve_{\ttime}(\ahyp,\swprice),\svaluemanuf_{\ttime}(\ahyp))_{\ttime\in\{1,\dots,\horizon\}}$, and to prove that they satisfy \cref{it:w*-bou-c>0,it:z*-bou-c>0,it:z>0-c>0,it:w<wmax-c>0}.
More precisely, we prove that, if at period $\ttime+1$ the value functions are well defined  and satisfy \cref{it:z>0-c>0} and \cref{it:w<wmax-c>0}, then it is possible to define strategies and  value functions that satisfy \cref{it:w*-bou-c>0,it:z*-bou-c>0,it:z>0-c>0,it:w<wmax-c>0} at period $\ttime$.

We start the induction at period $\horizon+1$, where the value functions are well defined and identically equal to $0$. 
Because $\retailprice>0$ and $\retailprice-\cost>0$,
\cref{it:z>0-c>0,it:w<wmax-c>0} are true for $\ttime=\horizon+1$.

We now assume that for some $ \ttime \leq \horizon$ and $\ahyp\in \Aset$,  $\svalueret_{\ttime+1}(\ahyp),\svalueretbreve_{\ttime+1}(\ahyp,\swprice),\svaluemanuf_{\ttime+1}(\ahyp)$ are well defined and satisfy \cref{it:z>0-c>0,it:w<wmax-c>0}:
\begin{align}
\label{eq:Induction-for_z>0}
\retailprice+ \svalueret_{\ttime + 1}(\ahyp+1)-\svalueret_{\ttime + 1}(\ahyp)
&>0,\\
\label{eq:Induction-for_w}
\retailprice-\cost+ \svalueret_{\ttime + 1}(\ahyp+1)+\svaluemanuf_{\ttime + 1}(\ahyp+1)
&>\svalueret_{\ttime + 1}(\ahyp)+\svaluemanuf_{\ttime + 1}(\ahyp).
\end{align}
We first construct $\sorder^{*}_{\ttime}(\ahyp,\swprice)$ and $\svalueretbreve_{\ttime}(\ahyp,\swprice)$ by analyzing the retailer’s payoff and choices.
Consider the expression in \cref{eq:standardized-retailer} and define
\begin{equation}
\label{eq:Reta-Payof-t-c>0}
\uvalueret_{\ttime}(\swprice,\sorder \mid \ahyp) \coloneqq \retailprice\preddemand(\sorder \mid \ahyp-1,1) - (\ahyp-1)\swprice\sorder  + \svalueret_{\ttime + 1}(\ahyp + 1) \preddemand(\sorder \mid \ahyp - 1,1) 
+ \svalueret_{\ttime + 1}(\ahyp) \survdemand (\sorder \mid \ahyp - 1,1). 
\end{equation}
If we define  
\begin{equation}
\label{eq:p-bar}
\overline{\retailprice}_{\ttime + 1}(\ahyp) \coloneqq \retailprice+\svalueret_{\ttime + 1}(\ahyp+1)-\svalueret_{\ttime + 1}(\ahyp),    
\end{equation}
and we use the equality 
\begin{equation}
\label{ewq:G-G-prime}   \preddemand'(\sorder \mid \ahyp-1,1)=(\ahyp-1)\survdemand (\sorder \mid \ahyp,1), 
\end{equation}
then we get 
\begin{equation}
\label{eq:partial-M-R}
\frac{\partial \uvalueret_{\ttime}(\swprice,\sorder \mid \ahyp)}{\partial \sorder}   
=
\left(\ahyp-1\right)\left[ \survdemand (\sorder \mid \ahyp,1)  \overline{\retailprice}_{\ttime + 1}(\ahyp)-\swprice \right].
\end{equation}
By the inductive hypothesis of \cref{it:z>0-c>0}, we know that $\overline{\retailprice}_{\ttime + 1}(\ahyp)>0$; hence, the partial derivative in \cref{eq:partial-M-R} is strictly decreasing in $\sorder$. 
Moreover, the limit of this derivative as $\sorder \to 0$ is $(\ahyp-1)\left[\overline{\retailprice}_{\ttime + 1}(\ahyp)-\swprice\right]$, whereas the limit as $\sorder \to +\infty$ is $-(\ahyp-1)\swprice$.

We have two cases. If $\swprice < \overline{\retailprice}_{\ttime + 1}(\ahyp)$, then the limit at $0$ of the above partial derivative is positive, and the equation $\partial \uvalueret_{\ttime}(\swprice,\sorder \mid \ahyp)/\partial \sorder = 0$ is equivalent to
\begin{equation}
\label{eq:optimal-order-t-c>0}
\survdemand(\sorder \mid \ahyp,1) = \frac{\swprice}{\overline{\retailprice}_{\ttime + 1}(\ahyp)},
\end{equation}
whose unique solution is
\begin{equation}
\label{eq:z*-t-c>0}
\sorder_{\ttime}^{*}(\ahyp,\swprice)  = \left(\frac{\overline{\retailprice}_{\ttime + 1}(\ahyp)}{\swprice}  \right)^{1/\ahyp} - 1.
\end{equation}

\noindent If $\swprice\geq \overline{\retailprice}_{\ttime + 1}(\ahyp)$, then the maximum is attained at zero. Summarizing both cases, we obtain 
\begin{equation}
\label{z*-t-c>0_global}
\sorder_{\ttime}^{*}(\ahyp,\swprice)=\max\left(\left(\frac{\overline{\retailprice}_{\ttime + 1}(\ahyp)}{\swprice}  \right)^{1/\ahyp} - 1,0\right), \quad \text{for all }\swprice \in \wprices.
\end{equation}
We can now define the retailer's value function at period $\ttime$ as 
$\svalueretbreve_{\ttime}(\ahyp,\swprice)\coloneqq \uvalueret_{\ttime}(\swprice,\sorder_{\ttime}^{*}(\ahyp,\swprice) \mid \ahyp).
$

We now turn to the manufacturer and prove that  $\swprice^{*}_{\ttime}(\ahyp)$ exists, is unique, and satisfies \cref{it:w*-bou-c>0}. 
Consider the expression in \cref{eq:standardized-manufacturer} and define
\begin{equation}
\label{eq:nor-pro-supl-t-c>0}
\uvaluemanuf_{\ttime}(\swprice\mid \ahyp) \coloneqq (\ahyp-1)(\swprice-\cost)\sorder^{*}_{\ttime}(\ahyp,\swprice)  + \svaluemanuf_{\ttime + 1}(\ahyp + 1) \preddemand(\sorder^{*}_{\ttime}(\ahyp,\swprice)\mid \ahyp - 1,1) + \svaluemanuf_{\ttime + 1}(\ahyp) \survdemand (\sorder^{*}_{\ttime}(\ahyp,\swprice)\mid \ahyp - 1,1). 
\end{equation}

If $\swprice\geq \overline{\retailprice}_{\ttime + 1}(\ahyp)$, then  \cref{z*-t-c>0_global,eq:nor-pro-supl-t-c>0} imply $\uvaluemanuf_{\ttime}(\swprice\mid \ahyp)=\svaluemanuf_{\ttime + 1}(\ahyp).$
If $ \swprice < \overline{\retailprice}_{\ttime + 1}(\ahyp)$, then, by \cref{ewq:G-G-prime},  
\begin{equation}
\label{eq:dev-nor-pro-supl-t-c>0}
\begin{split}
\frac{\partial \uvaluemanuf_{\ttime}(\swprice\mid \ahyp)}{\partial \swprice}   
&= (\ahyp-1) \left[\sorder^{*}_{\ttime}(\ahyp,\swprice)+   (\swprice-\cost)\frac{\partial \sorder^{*}_{\ttime}(\ahyp,\swprice)}{\partial \swprice}   +  \frac{\partial \sorder^{*}_{\ttime}(\ahyp,\swprice)}{\partial \swprice}\survdemand (\sorder^{*}_{\ttime}(\ahyp,\swprice)\mid \ahyp,1) \right.\\
&\quad\left.\left(\svaluemanuf_{\ttime + 1}(\ahyp + 1) -\svaluemanuf_{\ttime + 1}(\ahyp)\right)\right].
\end{split}
\end{equation}
We now study the behavior of $\swprice \mapsto \uvaluemanuf_{\ttime}(\swprice\mid \ahyp)$ on $(0, \overline{\retailprice}_{\ttime + 1}(\ahyp))$.  Differentiating \cref{eq:z*-t-c>0} with respect to $\swprice$, we obtain
\begin{equation}
\label{eq:dev_z-t}
\frac{\partial \sorder^{*}_{\ttime}(\ahyp,\swprice)}{\partial \swprice} 
= -\frac{\parens*{\overline{\retailprice}_{\ttime + 1}(\ahyp)/\swprice}^{1/\ahyp}}{\ahyp \swprice}=-\frac{1+\sorder^{*}_{\ttime}(\ahyp,\swprice)}{\ahyp \swprice}.
\end{equation}
Plugging  \cref{eq:dev_z-t} into \cref{eq:dev-nor-pro-supl-t-c>0}, we obtain 
\begin{equation}
\label{eq:partial-M-s-t}
\frac{\partial \uvaluemanuf_{\ttime}(\swprice\mid \ahyp)}{\partial \swprice}  
=  
(\ahyp-1)\bracks*{\left(1+ \sorder^{*}_{\ttime}(\ahyp,\swprice)\right)\parens*{1-\frac{\swprice-\cost}{\ahyp \swprice}-\frac{\survdemand (\sorder^{*}_{\ttime}(\ahyp,\swprice)\mid \ahyp,1)\parens*{\svaluemanuf_{\ttime + 1}(\ahyp + 1)
-\svaluemanuf_{\ttime + 1}(\ahyp)}}{\ahyp \swprice}}-1}.
\end{equation}
Because $\sorder_{\ttime}^{*}(\ahyp,\swprice)$ is the solution of \cref{eq:optimal-order-t-c>0}, plugging its value into \cref{eq:partial-M-s-t}, we obtain
\begin{equation}
\label{eq:Der-supli-t-c>0}
\frac{\partial \uvaluemanuf_{\ttime}(\swprice\mid \ahyp)}{\partial \swprice}   =   (\ahyp-1)\left[\left(1+ \sorder^{*}_{\ttime}(\ahyp,\swprice)\right)\left(1-\frac{1}{\ahyp}+\frac{\cost}{\ahyp \swprice}-\frac{\svaluemanuf_{\ttime + 1}(\ahyp + 1) -\svaluemanuf_{\ttime + 1}(\ahyp)}{\ahyp\overline{\retailprice}_{\ttime + 1}(\ahyp)}\right)-1\right].
\end{equation}
We show that $\partial \uvaluemanuf_{\ttime}(\swprice\mid \ahyp)/\partial \swprice=0$ has a unique solution $\swprice^{*}_{\ttime}(\ahyp)$ with $0<\swprice^{*}_{\ttime}(\ahyp)<\overline{\retailprice}_{\ttime+1}(\ahyp)$. 
Let
\begin{equation}
\label{eq:zeta_function}
\oneplus_{\ttime}(\ahyp,\swprice) \coloneqq 1  - \frac{1}{\ahyp} + \frac{\cost}{\ahyp \swprice} - \frac{\svaluemanuf_{\ttime + 1}(\ahyp + 1) - \svaluemanuf_{\ttime + 1}(\ahyp)}
{\ahyp \overline{\retailprice}_{\ttime + 1}(\ahyp)}.
\end{equation}
Then \cref{eq:Der-supli-t-c>0} becomes 
\begin{equation}
\label{eq:Der2-supli-t-c>0}
\frac{\partial \uvaluemanuf_{\ttime}(\swprice\mid \ahyp)}{\partial \swprice}   
= (\ahyp-1)\left[ \left(1+ \sorder^{*}_{\ttime}(\ahyp,\swprice)\right)\oneplus_{\ttime}(\ahyp,\swprice)-1\right].
\end{equation}
This partial derivative is continuous in $\swprice$, because it is a sum and product of continuous functions.

By the inductive  hypothesis \cref{eq:Induction-for_z>0} and using \cref{eq:z*-t-c>0}, we have that $ \lim_{\swprice \searrow 0}{\sorder^{*}_{\ttime}(\ahyp,\swprice)}=+\infty$, and $ \lim_{\swprice \searrow 0}{\oneplus_{\ttime}(\ahyp,\swprice)}=+\infty$. 
It follows that $ \lim_{\swprice \searrow 0}\partial \uvaluemanuf_{\ttime}(\swprice\mid \ahyp)/\partial \swprice=+\infty$. 
Moreover, by \cref{eq:z*-t-c>0}, $\sorder^{*}_{\ttime}\parens*{\ahyp,\overline{\retailprice}_{\ttime + 1}(\ahyp)}=0$ and after some algebra, we obtain
\begin{align}
\oneplus_{\ttime}(\ahyp,\overline{\retailprice}_{\ttime + 1}(\ahyp))-1& =\frac{\cost-\overline{\retailprice}_{\ttime + 1}(\ahyp)-\svaluemanuf_{\ttime + 1}(\ahyp + 1)+ \svaluemanuf_{\ttime + 1}(\ahyp)}{\ahyp\overline{\retailprice}_{\ttime + 1}(\ahyp)}\\
\label{eq:zeta-a-pa}
&
=\frac{\cost-\retailprice- \svalueret_{\ttime + 1}(\ahyp + 1)+\svalueret_{\ttime + 1}(\ahyp)-\svaluemanuf_{\ttime + 1}(\ahyp + 1)+ \svaluemanuf_{\ttime + 1}(\ahyp)}{\ahyp\overline{\retailprice}_{\ttime + 1}(\ahyp)}.
\end{align}
In \cref{eq:zeta-a-pa}, the denominator is strictly positive by \cref{eq:Induction-for_z>0}, and the numerator is strictly negative by \cref{eq:Induction-for_w}. 
It follows that
$\oneplus_{\ttime}\left(\ahyp,\overline{\retailprice}_{\ttime + 1}(\ahyp)\right)-1<0$; therefore
$\partial \uvaluemanuf_{\ttime}(\overline{\retailprice}_{\ttime + 1}(\ahyp)\mid \ahyp)/
\partial \swprice<0$.
Together, the two boundary results and continuity imply that the equation 
$\partial \uvaluemanuf_{\ttime}(\swprice\mid \ahyp)/\partial \swprice=0$ has at least one solution in
$\left(0,\overline{\retailprice}_{\ttime+1}(\ahyp)\right).$

It remains to show that this equation cannot have more than one solution. By \cref{eq:z*-t-c>0}, the function $\sorder^{*}_{\ttime}(\ahyp,\swprice)$ is strictly decreasing in  $\swprice$ and positive for $\swprice\in (0, \overline{\retailprice}_{\ttime + 1}(\ahyp))$, with $\sorder^{*}_{\ttime}\parens*{\ahyp,\overline{\retailprice}_{\ttime + 1}(\ahyp)}=0$. From \cref{eq:zeta_function}, the function $\oneplus_{\ttime}(\ahyp,\swprice)$ is strictly decreasing for $\swprice\in\wprices$. There are two cases: Either $\oneplus_{\ttime}(\ahyp,\swprice)$ is strictly positive on 
$(0,\overline{\retailprice}_{\ttime + 1}(\ahyp))$, or there exists 
$\varepsilon \in (0,\overline{\retailprice}_{\ttime + 1}(\ahyp)]$ such that 
$\oneplus_{\ttime}(\ahyp,\swprice)$ is strictly positive on $(0,\varepsilon)$ and nonpositive on 
$[\varepsilon,\overline{\retailprice}_{\ttime + 1}(\ahyp)]$.

In the first case, by \cref{eq:Der2-supli-t-c>0}, the partial derivative 
$\partial \uvaluemanuf_{\ttime}(\swprice\mid \ahyp)/\partial \swprice$ is strictly decreasing in $\swprice$ on 
$(0,\overline{\retailprice}_{\ttime + 1}(\ahyp)]$ because 
$\left(1+\sorder^{*}_{\ttime}(\ahyp,\swprice)\right)\oneplus_{\ttime}(\ahyp,\swprice)$ is the product of two functions, 
$1+\sorder^{*}_{\ttime}(\ahyp,\swprice)$ and $\oneplus_{\ttime}(\ahyp,\swprice)$, that are both strictly decreasing in $\swprice$ and strictly positive. 
Hence, the first-order condition admits at most one solution.

In the second case, the preceding argument applies on $(0,\varepsilon)$; hence, the first-order condition admits at most one solution on $(0,\varepsilon)$. Moreover, for 
$\swprice\in[\varepsilon,\overline{\retailprice}_{\ttime + 1}(\ahyp)]$, we have 
$\oneplus_{\ttime}(\ahyp,\swprice)\leq 0$ and 
$1+\sorder^{*}_{\ttime}(\ahyp,\swprice)>0$. Therefore, by \cref{eq:Der2-supli-t-c>0}, $\partial \uvaluemanuf_{\ttime}(\swprice\mid \ahyp)/\partial \swprice
\leq -(\ahyp-1)<0$, because $\ahyp>1$. 
Thus, the first-order condition cannot be satisfied on 
$[\varepsilon,\overline{\retailprice}_{\ttime + 1}(\ahyp)]$.

This proves \cref{it:w*-bou-c>0} at period $\ttime$.
Because $\swprice^{*}_{\ttime}(\ahyp)<\overline{\retailprice}_{\ttime + 1}(\ahyp)$, using \cref{z*-t-c>0_global}, we get 
\begin{equation}
\label{eq:z-closef-form}
\sordertwo^{*}_{\ttime}(\ahyp)=\sorder^{*}_{\ttime}\left(\ahyp,\swprice^{*}_{\ttime}(\ahyp)\right)=\left(\frac{\overline{\retailprice}_{\ttime + 1}(\ahyp)}{\swprice^{*}_{\ttime}(\ahyp)}\right)^{1/\ahyp}-1 \in (0,\infty),
\end{equation}
proving \cref{it:z*-bou-c>0} at period $\ttime$. 
Moreover, one can define
\begin{equation}
\label{eq:value-functions}    
\svalueret_{\ttime}(\ahyp) 
\coloneqq  \uvalueret_{\ttime}\bigl(\swprice^{*}_{\ttime}(\ahyp), \sordertwo_{\ttime}^{*}(\ahyp) \mid \ahyp \bigr),    \text{ and } \svaluemanuf_{\ttime}(\ahyp) 
\coloneqq 
\uvaluemanuf_{\ttime}\bigl(\swprice^{*}_{\ttime}(\ahyp)
\mid \ahyp \bigr).
\end{equation}
We now prove that the above value functions satisfy \cref{it:z>0-c>0,it:w<wmax-c>0}. 
From \cref{eq:Reta-Payof-t-c>0}, one has
\begin{equation}
\label{eq:ret_ineq-t-c>0}
\svalueret_{\ttime}(\ahyp)  
= \retailprice - \swprice^{*}_{\ttime}(\ahyp)\sordertwo_{\ttime}^{*}(\ahyp)(\ahyp-1) - \survdemand (\sordertwo_{\ttime}^{*}(\ahyp)\mid \ahyp-1,1)\overline{\retailprice}_{\ttime + 1}(\ahyp)+\svalueret_{\ttime + 1}(\ahyp+1).
\end{equation}
Choosing $\sorder_{\ttime}(\ahyp+1,\swprice)=0$, we can prove that, for any  $\swprice\in\wprices$, we have $\svalueret_{\ttime + 1}(\ahyp+1) \le \svalueret_{\ttime}(\ahyp+1)$.
It follows that 
\begin{align}
\svalueret_{\ttime}(\ahyp)  &\leq \retailprice - \swprice^{*}_{\ttime}(\ahyp)\sordertwo_{\ttime}^{*}(\ahyp)(\ahyp-1) - \survdemand (\sordertwo_{\ttime}^{*}(\ahyp)\mid \ahyp-1,1)\overline{\retailprice}_{\ttime + 1}(\ahyp)+\svalueret_{\ttime}(\ahyp+1).
\end{align}
Rearranging and using the inductive  hypothesis \cref{eq:Induction-for_z>0,it:z*-bou-c>0}, we obtain 
\begin{equation}
\retailprice+\svalueret_{\ttime}(\ahyp+1)-\svalueret_{\ttime}(\ahyp)\geq \swprice^{*}_{\ttime}(\ahyp)\sordertwo_{\ttime}^{*}(\ahyp)(\ahyp-1) + \survdemand (\sordertwo_{\ttime}^{*}(\ahyp)\mid \ahyp-1,1)\overline{\retailprice}_{\ttime + 1}(\ahyp)>0,
\end{equation}
which proves \cref{it:z>0-c>0} at period $\ttime$.

From \cref{eq:nor-pro-supl-t-c>0}, one has
\begin{equation}
\label{eq:norma_sup_pay-t}
\svaluemanuf_{\ttime}(\ahyp) 
= (\ahyp-1)(\swprice^{*}_{\ttime}(\ahyp)-\cost)\sordertwo_{\ttime}^{*}(\ahyp)  + \svaluemanuf_{\ttime + 1}(\ahyp + 1) \preddemand(\sordertwo_{\ttime}^{*}(\ahyp)\mid \ahyp - 1,1) + \svaluemanuf_{\ttime + 1}(\ahyp) \survdemand (\sordertwo_{\ttime}^{*}(\ahyp)\mid \ahyp - 1,1).
\end{equation}
From \cref{eq:p-bar,eq:ret_ineq-t-c>0}, the retailer's  standardized value function can be written as
\begin{equation}
\label{eq:norma_ret_pay-t}
\begin{split}
\svalueret_{\ttime}(\ahyp) 
&= \retailprice\preddemand(\sordertwo_{\ttime}^{*}(\ahyp)\mid \ahyp-1,1) - (\ahyp-1)\swprice^{*}_{\ttime}(\ahyp)\sordertwo_{\ttime}^{*}(\ahyp)  + \svalueret_{\ttime + 1}(\ahyp + 1) \preddemand(\sordertwo_{\ttime}^{*}(\ahyp)\mid \ahyp - 1,1)  \\
&\quad + \svalueret_{\ttime + 1}(\ahyp) \survdemand (\sordertwo_{\ttime}^{*}(\ahyp)\mid \ahyp - 1,1).
\end{split}
\end{equation}
Summing \cref{eq:norma_sup_pay-t,eq:norma_ret_pay-t}, we get
\begin{equation}
\label{eq:sum_ineq-t,c>0}
\begin{split}
\svaluemanuf_{\ttime}(\ahyp)+ \svalueret_{\ttime}(\ahyp) 
&= - (\ahyp-1)\cost\sordertwo_{\ttime}^{*}(\ahyp)+\preddemand(\sordertwo_{\ttime}^{*}(\ahyp)\mid \ahyp - 1,1)\left(\retailprice+\svalueret_{\ttime + 1}(\ahyp + 1)+\svaluemanuf_{\ttime + 1}(\ahyp + 1)\right)\\
& \quad + \survdemand (\sordertwo_{\ttime}^{*}(\ahyp)\mid \ahyp - 1,1)\left(\svalueret_{\ttime + 1}(\ahyp )+\svaluemanuf_{\ttime + 1}(\ahyp )\right). 
\end{split}
\end{equation}
\cref{eq:Induction-for_w}  implies
\begin{equation}
\begin{split}
\svaluemanuf_{\ttime}(\ahyp)+ \svalueret_{\ttime}(\ahyp) 
&< - (\ahyp-1)\cost\sordertwo_{\ttime}^{*}(\ahyp)+\preddemand(\sordertwo_{\ttime}^{*}(\ahyp)\mid \ahyp - 1,1)\left(\retailprice+\svalueret_{\ttime + 1}(\ahyp + 1)+\svaluemanuf_{\ttime + 1}(\ahyp + 1)\right)\\
& \quad + \survdemand (\sordertwo_{\ttime}^{*}(\ahyp)\mid \ahyp - 1,1)\left(\retailprice-\cost+\svalueret_{\ttime + 1}(\ahyp + 1)+\svaluemanuf_{\ttime + 1}(\ahyp + 1)\right) \\
    &= - (\ahyp-1)\cost\sordertwo_{\ttime}^{*}(\ahyp)+\retailprice+\svalueret_{\ttime + 1}(\ahyp + 1)+\svaluemanuf_{\ttime + 1}(\ahyp + 1)-\cost\survdemand (\sordertwo_{\ttime}^{*}(\ahyp)\mid \ahyp - 1,1),    
\end{split}    
\end{equation}
where the equality follows from $\preddemand(\sordertwo_{\ttime}^{*}(\ahyp)\mid \ahyp - 1,1)+\survdemand(\sordertwo_{\ttime}^{*}(\ahyp)\mid \ahyp - 1,1)=1$, and some algebra. Equivalently,
\begin{equation}
\label{eq:sum_ineq}
\svaluemanuf_{\ttime}(\ahyp)+\svalueret_{\ttime}(\ahyp)<\cost\preddemand(\sordertwo_{\ttime}^{*}(\ahyp)\mid \ahyp - 1,1)- (\ahyp-1)\cost \sordertwo_{\ttime}^{*}(\ahyp)+\retailprice-\cost+\svalueret_{\ttime + 1}(\ahyp + 1)+\svaluemanuf_{\ttime + 1}(\ahyp + 1). 
\end{equation}
Consider the function $\auxf(\sorder) \coloneqq \cost\preddemand(\sorder \mid \ahyp - 1,1)- (\ahyp-1)\cost \sorder.$
By choosing $\swprice \geq \overline{\retailprice}_{\ttime+1}(\ahyp+1)$, we have that $\svaluemanuf_{\ttime+1} (\ahyp+1) \leq \svaluemanuf_{\ttime}(\ahyp+1)$.  
Combining this with the inequality $\svalueret_{\ttime+1}(\ahyp+1)\leq \svalueret_{\ttime}(\ahyp+1)$, \cref{eq:sum_ineq} implies
\begin{equation}\label{eq:order}
\svalueret_{\ttime}(\ahyp)+\svaluemanuf_{\ttime}(\ahyp)< \auxf(\sordertwo_{\ttime}^{*}(\ahyp))+\retailprice-\cost+\svalueret_{\ttime}(\ahyp + 1)+\svaluemanuf_{\ttime}(\ahyp+ 1). 
\end{equation}
We have $\auxf(0)=0$ and, using \cref{ewq:G-G-prime},
\begin{equation}
\label{eq:der-psi}
\auxf'(\sorder)
= \cost \preddemand'(\sorder \mid \ahyp - 1,1)- (\ahyp-1)\cost
=\cost(\ahyp-1)(\survdemand (\sorder \mid \ahyp,1)-1).
\end{equation}
For every $\sorder>0$, $\survdemand(\sorder \mid \ahyp,1)<1$, and $\survdemand'(\sorder \mid \ahyp,1)<0$. 
We conclude that $\auxf(\sorder) < 0$, for all $\sorder>0$; therefore \cref{eq:order} implies \cref{it:w<wmax-c>0} at period $\ttime$.
\end{proof}

\begin{proof}[Proof of \cref{th:exponential}]
We first prove \ref{it:th:exponential-a}. From \cref{prop:Norma-properties-exp-c>0}, there exists a unique \ac{SMPS} $(\swprice^{*}_{\ttime}(\ahyp), \sorder^{*}_{\ttime}(\ahyp,\swprice))$ of $\contract$. 
Plugging $\overline{\retailprice}_{\ttime + 1}(\argdot)$ from \cref{eq:p-bar} into \cref{z*-t-c>0_global}, we get  \cref{eq:z-t^{*}(a)_c>0}. 
Moreover, \cref{eq:optimal-standa-order} was already established in the proof of \cref{prop:Norma-properties-exp-c>0} (see \cref{eq:z-closef-form}). 

Given that  $\swprice^{*}_{\ttime}(\ahyp)$ satisfies the first-order conditions, \cref{eq:Der-supli-t-c>0} can be rewritten as 
\begin{equation}
\label{eq:manufacturer_foc-t-c>0}
\left(1+ \sordertwo_{\ttime}^{*}(\ahyp)\right)\left[1+\frac{\cost}{\ahyp \swprice^{*}_{\ttime}(\ahyp)}-\frac{1}{\ahyp}-\frac{\svaluemanuf_{\ttime + 1}(\ahyp + 1)-\svaluemanuf_{\ttime + 1}(\ahyp)}{\ahyp
(\retailprice+\svalueret_{\ttime + 1}(\ahyp + 1) -\svalueret_{\ttime + 1}(\ahyp))}\right]=1.
\end{equation}
From \cref{eq:predictive-distr-func}, we see that \cref{eq:optimal-order-t-c>0} is equivalent to
\begin{align}\label{eq:alt_z}
\left(1+\sordertwo_{\ttime}^{*}(\ahyp)\right)^\ahyp= \frac{\overline{\retailprice}_{\ttime + 1}(\ahyp)}{\swprice_{\ttime}^{*}(\ahyp)}.
\end{align}
Raising both sides of \cref{eq:manufacturer_foc-t-c>0} to the power $\ahyp$ and using \cref{eq:alt_z}, we obtain \cref{eq:optimal_wholes_exp_c>0}.  

We now focus on the standardized value functions. 
First we prove the following equalities:
\begin{align}
\label{eq:difference_retailer}\svalueret_{\ttime}(\ahyp) 
& = \retailprice - \swprice^{*}_{\ttime}(\ahyp )\left(1+\ahyp\sordertwo^{*}_{\ttime }(\ahyp )\right)+\svalueret_{\ttime + 1}(\ahyp+1),  \\
\label{eq:difference_manufacturer}
\svaluemanuf_{\ttime}(\ahyp)
&= \swprice^{*}_{\ttime}(\ahyp)-\cost(1+\ahyp \sordertwo_{\ttime}^{*}(\ahyp)) + \svaluemanuf_{\ttime + 1}(\ahyp+1).
\end{align}
Recall that \cref{eq:ret_ineq-t-c>0} is
\begin{equation}
\label{eq:ret_val_a-1}
\svalueret_{\ttime}(\ahyp) = \retailprice - (\ahyp-1)\swprice^{*}_{\ttime}(\ahyp)\sordertwo_{\ttime}^{*}(\ahyp) - \survdemand (\sordertwo_{\ttime}^{*}(\ahyp)\mid \ahyp-1,1)\overline{\retailprice}_{\ttime + 1}(\ahyp)+\svalueret_{\ttime + 1}(\ahyp+1).
\end{equation}
From the first-order conditions, we write \cref{eq:optimal-order-t-c>0} as
\begin{equation}
\label{eq:ret_foc_a-1-c>0}
\survdemand (\sordertwo_{\ttime}^{*}(\ahyp)\mid \ahyp-1,1)\overline{\retailprice}_{\ttime + 1}(\ahyp) =\swprice_{\ttime}^{*}(\ahyp)(1+\sordertwo_{\ttime}^{*}(\ahyp)).
\end{equation}
Substituting \cref{eq:ret_foc_a-1-c>0} into \cref{eq:ret_val_a-1}, one obtains \cref{eq:difference_retailer}. 

The manufacturer’s  value function in \cref{eq:norma_sup_pay-t} can be written as
\begin{equation}
\label{eq:v-S-expression}
    \svaluemanuf_{\ttime}(\ahyp) = (\ahyp-1)(\swprice^{*}_{\ttime}(\ahyp)-\cost)\sordertwo_{\ttime}^{*}(\ahyp) +\survdemand (\sordertwo_{\ttime}^{*}(\ahyp)\mid \ahyp - 1,1) (\svaluemanuf_{\ttime + 1}(\ahyp)-\svaluemanuf_{\ttime + 1}(\ahyp + 1)) + \svaluemanuf_{\ttime + 1}(\ahyp+1). 
\end{equation}
Substituting \cref{eq:ret_foc_a-1-c>0} into \cref{eq:v-S-expression}, we get
\begin{equation}
\label{eq:Manufacturer_norm_prof_c>0}
\svaluemanuf_{\ttime}(\ahyp) = (\ahyp-1)(\swprice^{*}_{\ttime}(\ahyp)-\cost)\sordertwo_{\ttime}^{*}(\ahyp) - \swprice^{*}_{\ttime}(\ahyp)(1+\sordertwo_{\ttime}^{*}(\ahyp))\frac{\svaluemanuf_{\ttime + 1}(\ahyp+1)-\svaluemanuf_{\ttime + 1}(\ahyp )}{\retailprice + \svalueret_{\ttime + 1}(\ahyp+1) - \svalueret_{\ttime + 1}(\ahyp)} + \svaluemanuf_{\ttime + 1}(\ahyp+1). 
\end{equation}
Multiplying both sides of  \cref{eq:manufacturer_foc-t-c>0} by $\ahyp \swprice_{\ttime}^{*}(\ahyp)$, we obtain
\begin{equation}
\left(1+ \sordertwo_{\ttime}^{*}(\ahyp)\right)\left[\ahyp \swprice^{*}_{\ttime}(\ahyp)+\cost-\swprice^{*}_{\ttime}(\ahyp)-\swprice^{*}_{\ttime}(\ahyp)\frac{\svaluemanuf_{\ttime + 1}(\ahyp + 1)-\svaluemanuf_{\ttime + 1}(\ahyp)}{\retailprice+\svalueret_{\ttime + 1}(\ahyp + 1) -\svalueret_{\ttime + 1}(\ahyp)}\right]=\ahyp \swprice^{*}_{\ttime}(\ahyp),
\end{equation}
or, equivalently,
\begin{equation}\label{eq:sup_foc_exp_c>0}
\swprice^{*}_{\ttime}(\ahyp)\left(1+ \sordertwo_{\ttime}^{*}(\ahyp)\right)\frac{\svaluemanuf_{\ttime + 1}(\ahyp + 1)-\svaluemanuf_{\ttime + 1}(\ahyp)}{\retailprice+\svalueret_{\ttime + 1}(\ahyp + 1) -\svalueret_{\ttime + 1}(\ahyp)}=\left(1+ \sordertwo_{\ttime}^{*}(\ahyp)\right)\left[\ahyp \swprice^{*}_{\ttime}(\ahyp)+\cost-\swprice^{*}_{\ttime}(\ahyp)\right]-\ahyp \swprice^{*}_{\ttime}(\ahyp).
\end{equation}
Substituting the left-hand side of  \cref{eq:sup_foc_exp_c>0} into  \cref{eq:Manufacturer_norm_prof_c>0} and simplifying, we obtain \cref{eq:difference_manufacturer}.

The telescoping series implied by \cref{eq:difference_retailer,eq:difference_manufacturer}, and the terminal conditions  $\svalueret_{\horizon+1}(\ahyp)=\svaluemanuf_{\horizon+1}(\ahyp)=0$ yield \cref{eq:fr-t_exp,eq:fs-t_exp}.

We now prove \ref{it:th:exponential-b}. 
Because the \ac{SMPS}
$\parens*{\swprice_{\ttime}^{*}(\ahyp),
\sorder_{\ttime}^{*}(\ahyp,\wholesaleprice)}_{1\le \ttime\le \horizon}$
exists and is unique, \cref{pr:unique-SMPS} implies the existence and uniqueness of the corresponding
\ac{MPE} of $\gameenv$.

Moreover, because $\lnewsv(\yvar)=\yvar$ corresponds to $\weibull=1$, the scaling relation
in \cref{eq:traduction} gives, for every
$\ttime\in\braces*{1,\dots,\horizon}$, $\ahyp\in\Aset$,
$\bhyp\in\Bset$, and $\wholesaleprice\in\wprices$,
\[
\wholesaleprice_{\ttime}^{*}(\ahyp,\bhyp)
=
\swprice_{\ttime}^{*}(\ahyp),
\qquad
\order_{\ttime}^{*}(\ahyp,\bhyp,\wholesaleprice)
=
\bhyp\,\sorder_{\ttime}^{*}(\ahyp,\wholesaleprice).
\]
This is exactly \cref{eq:MPE-SMPS-exponential}, which proves
\ref{it:th:exponential-b}.
\end{proof}

%
%
%
%

\section{List of Symbols}
\label{se:symbols}

\begin{longtable}{p{.18\textwidth} p{.77\textwidth}}

$\ahyp_{\ttime}$ & shape hyperparameter of the posterior distribution at time $\ttime$\\

$\Aset$ & $\braces*{\ahyp_{1}+\nuncens\colon \nuncens\in\naturals}$\\

$\bhyp_{\ttime}$ & rate hyperparameter of the posterior distribution at time $\ttime$\\
$\Bset$ & $[\bhyp_{1},\infty)$\\
$\nuncens_{\ttime}$ & number of uncensored observations up to period $\ttime$\\

$\cost$ & manufacturer’s per-unit production cost\\

$\demand_{\ttime}$ & realized demand at time $\ttime$\\
$\Demand_{\ttime}$ & random demand at time $\ttime$\\

$\Expect$ & expectation operator\\
$\Expectgen$ & standardized continuation-value operator, defined before \cref{de:standardized}\\

$\ffunc$ &  function used in the definition of $\Expectgen$\\
$\distrdemand(\argdot\mid\parameter)$ & conditional demand distribution function given $\Parameter=\parameter$\\
$\preddemand(\argdot\mid\ahyp,\bhyp)$ & predictive demand distribution function given $(\ahyp,\bhyp)$\\
$\survdemand(\argdot\mid\ahyp,\bhyp)$ & predictive demand survival function, defined by $\survdemand(\argdot\mid\ahyp,\bhyp)\coloneqq 1-\preddemand(\argdot\mid\ahyp,\bhyp)$\\
$\gameenv$ & dynamic wholesale-price game induced by $\contract$\\
$\Gr(\maximiz)$ & graph of the correspondence $\maximiz$\\

$\hvalueret_{\ttime}$ & scaled retailer value function used in the proof of \cref{th:main}, defined in \cref{reduction-formulas}\\
$\hvalueretbreve_{\ttime}$ & scaled retailer value function after observing the wholesale price, used in the proof of \cref{th:main}, defined in \cref{reduction-formulas}\\
$\hvaluemanuf_{\ttime}$ & scaled manufacturer value function used in the proof of \cref{th:main}, defined in \cref{reduction-formulas}\\

$\weibull$ & Weibull shape parameter; $\lnewsv(\yvar)=\yvar^{\weibull}$ in the Weibull case\\

$\lnewsv$ & function defining the newsvendor demand family in \cref{eq:demand-distribution}\\

$\manufacturer$ & manufacturer\\

$\naturals$ & set of nonnegative integers, $\braces*{0,1,2,\ldots}$\\

$\retailprice$ & exogenous retail price\\
$\overline{\retailprice}_{\ttime+1}(\ahyp)$ & effective retail price in the exponential-demand recursion, defined in \cref{eq:p-bar}\\
$\Prob$ & probability operator\\

$\order_{\ttime}$ & order quantity at time $\ttime$\\
$\order_{\ttime}^{*}$ & equilibrium order quantity at time $\ttime$\\
$\order^{\circ}$ & myopic optimal order quantity in the original problem, defined in \cref{eq:myopic-order-closed}\\
$\sorder_{\ttime}$ & standardized order quantity at time $\ttime$\\
$\scaleorder_{\ttime}$ & scaled order strategy used in the proof of \cref{th:main}\\
$\sorder_{\ttime}^{*}$ & equilibrium standardized retailer best response at time $\ttime$\\
$\sordertwo_{\ttime}^{*}$ & equilibrium standardized order quantity at time $\ttime$, after substituting the equilibrium wholesale price\\
$\sorder^{\circ}$ & myopic optimal standardized order quantity, defined in \cref{eq:myopic-standardized-order}\\
$\overline{\sorder}_{\ttime}(\ahyp)$ & upper bound used to restrict the retailer's action space in the proof of \cref{th:main}\\
$\orders$ & set of admissible order quantities\\
$\sorders$ & set of admissible standardized order quantities\\

$\retailer$ & retailer\\
$\reals$ & set of real numbers\\
$\reals_{+}$ & $[0,\infty)$\\
$\reals_{++}$ & $(0,\infty)$\\

$\sale_{\ttime}$ & realized sales at time $\ttime$\\
$\Sale_{\ttime}$ & random sales at time $\ttime$\\

$\ttime$ & discrete time period\\
$\horizon$ & time horizon\\

$\uvaluemanufsmall_{\ttime}$ & manufacturer’s auxiliary objective before leader-favorable tie-breaking, defined in \cref{eq:sup_Payof-t-c>0-toto}\\
$\uvalue$ & generic objective function used in \cref{suse:Berge}\\
$\uvalueret_{\ttime}$ & retailer’s auxiliary objective used in the proofs; see \cref{eq:M-r_exist,eq:Reta-Payof-t-c>0}\\
$\uvaluemanuf_{\ttime}$ & manufacturer’s auxiliary objective used in the proofs; see \cref{eq:selection,eq:nor-pro-supl-t-c>0}\\

$\svalue$ & generic value function used in \cref{suse:Berge}\\
$\svalueret_{\ttime}$ & retailer’s standardized value function at time $\ttime$, defined in \cref{de:standardized}\\
$\svalueretbreve_{\ttime}$ & retailer’s standardized value function after observing the wholesale price at time $\ttime$, defined in \cref{de:standardized}\\
$\svaluemanuf_{\ttime}$ & manufacturer’s standardized value function at time $\ttime$, defined in \cref{de:standardized}\\
$\Valueret_{\ttime}$ & retailer’s value function at time $\ttime$, defined in \cref{de:belief}\\
$\Valueretbreve_{\ttime}$ & retailer’s value function after observing the wholesale price at time $\ttime$, defined in \cref{de:belief}\\
$\Valuemanuf_{\ttime}$ & manufacturer’s value function at time $\ttime$, defined in \cref{de:belief}\\

$\underline{\wholesaleprice}$ & lower bound of the compact wholesale-price set in \cref{th:main}\\
$\overline{\wholesaleprice}$ & upper bound of the compact wholesale-price set in \cref{th:main}\\
$\wholesaleprice_{\ttime}$ & wholesale price at time $\ttime$\\
$\wholesaleprice_{\ttime}^{*}$ & equilibrium wholesale price at time $\ttime$\\
$\wholesaleprice^{\circ}$ & myopic optimal wholesale price in the original problem, defined in \cref{eq:myopic-wholesale}\\
$\swprice_{\ttime}$ & standardized wholesale price at time $\ttime$\\
$\scalewholesaleprice_{\ttime}$ & scaled wholesale-price strategy used in the proof of \cref{th:main}\\
$\swprice_{\ttime}^{*}$ & equilibrium standardized wholesale price at time $\ttime$\\
$\swprice^{\circ}$ & myopic optimal standardized wholesale price, defined in \cref{eq:myopic-wholesale}\\
$\wprices$ & set of admissible wholesale prices\\

$\Gamma(\argdot)$ & Gamma function\\

$\maximiz$ &  correspondence used in \cref{suse:Berge}\\

$\oneplus_{\ttime}$ & auxiliary function used in the proof of \cref{prop:Norma-properties-exp-c>0}, defined in \cref{eq:zeta_function}\\

$\parameter$ & realization of $\Parameter$\\
$\Parameter$ & unknown demand parameter\\

$\corresp$ &  correspondence used in \cref{suse:Berge}\\
$\Xi_{\ttime}(\ahyp)$ & compact retailer action set used in the proof of \cref{th:main}\\

$\spayoffretailer$ & retailer’s standardized per-period expected profit, defined in \cref{eq:standardized-payoff-retailer}\\
$\spayoffmanufacturer$ & manufacturer’s standardized per-period profit, defined in \cref{eq:standardized-payoff-manufacturer}\\
$\payoffretailer$ & retailer’s expected per-period profit, defined in \cref{eq:expected-per-period-profit-retailer}\\
$\payoffmanufacturer$ & manufacturer’s per-period profit, defined in \cref{eq:expected-per-period-profit-manufacturer}\\

$\prior$ & common prior density of $\Parameter$\\
$\transition$ & belief-transition function, defined in \cref{eq:Transi_funnction}\\

$\auxf$ & auxiliary function used in the proof of \cref{prop:Norma-properties-exp-c>0}\\
$\contract$ & model environment, $(\Parameter,\prior,(\distrdemand(\argdot\mid\parameter))_{\parameter},\retailprice,\cost,\horizon,\wprices,\orders)$\\

$\ind_{\{\argdot\}}$ & indicator function\\

\end{longtable}

%
%
%
%

\subsection*{Acknowledgments}

Marco Scarsini and Xavier Venel are members of GNAMPA-INdAM.
This work was partially supported by the  MIUR PRIN 2022EKNE5K ``Learning in markets and society,'' and the European Union-Next Generation EU Grant P2022XT8C8, component M4C2,
investment 1.1.

\bibliographystyle{apalike} 

\bibliography{Bibliog.bib}

\begin{thebibliography}{}

\bibitem[Aliprantis and Border, 2006]{AliBor:SPRINGER2006}
Aliprantis, C.~D. and Border, K.~C. (2006).
\newblock {\em Infinite Dimensional Analysis. A Hitchhiker's Guide}.
\newblock Springer, Berlin, third edition.

\bibitem[Anand et~al., 2008]{AnaAnuBas:MS2008}
Anand, K., Anupindi, R., and Bassok, Y. (2008).
\newblock Strategic inventories in vertical contracts.
\newblock {\em Management Sci.}, 54(10):1792--1804.

\bibitem[Azoury, 1985]{Azo:MS1985}
Azoury, K.~S. (1985).
\newblock Bayes solution to dynamic inventory models under unknown demand
  distribution.
\newblock {\em Management Sci.}, 31(9):1150--1160.

\bibitem[Bensoussan et~al., 2025]{BenSetWan:OR2025}
Bensoussan, A., Sethi, S., and Wang, S. (2025).
\newblock Technical note---{A} stationary infinite-horizon supply contract
  under asymmetric inventory information.
\newblock {\em Oper. Res.}, 73(1):270--277.

\bibitem[Berge, 1959]{Ber:DUNOD1959}
Berge, C. (1959).
\newblock {\em Espaces topologiques: {F}onctions multivoques}, volume Vol. III
  of {\em Collection Universitaire de Math\'ematiques}.
\newblock Dunod, Paris.

\bibitem[Besbes et~al., 2022]{BesChaMoaSS2022}
Besbes, O., Chaneton, J.~M., and Moallemi, C.~C. (2022).
\newblock The exploration-exploitation trade-off in the newsvendor problem.
\newblock {\em Stoch. Syst.}, 12(4):319--339.

\bibitem[Bhaskar and Roketskiy, 2023]{BhaskarRoketskiy:JET2023}
Bhaskar, V. and Roketskiy, N. (2023).
\newblock The ratchet effect: a learning perspective.
\newblock {\em J. Econom. Theory}, 214:Paper No. 105762, 22.

\bibitem[Bisi et~al., 2011]{BisDadTok:MSOM2011}
Bisi, A., Dada, M., and Tokdar, S. (2011).
\newblock A censored-data multiperiod inventory problem with newsvendor demand
  distributions.
\newblock {\em Manufacturing Service Oper. Management}, 13(4):525--533.

\bibitem[Braden and Freimer, 1991]{BraFre:MS1991}
Braden, D.~J. and Freimer, M. (1991).
\newblock Informational dynamics of censored observations.
\newblock {\em Management Sci.}, 37(11):1390--1404.

\bibitem[Breton et~al., 1988]{BreAljHau:JOTA1988}
Breton, M., Alj, A., and Haurie, A. (1988).
\newblock Sequential {S}tackelberg equilibria in two-person games.
\newblock {\em J. Optim. Theory Appl.}, 59(1):71--97.

\bibitem[Cachon, 2003]{Cac:HORMS2003}
Cachon, G.~P. (2003).
\newblock Supply chain coordination with contracts.
\newblock In {\em Supply Chain Management: Design, Coordination and Operation},
  pages 227--339. Elsevier.

\bibitem[Cesa-Bianchi et~al., 2023]{CesCesOsoScaWas:IFAAMS2023}
Cesa-Bianchi, N., Cesari, T., Osogami, T., Scarsini, M., and Wasserkrug, S.
  (2023).
\newblock Learning the {S}tackelberg equilibrium in a newsvendor game.
\newblock AAMAS '23, pages 242--250, Richland, SC. International Foundation for
  Autonomous Agents and Multiagent Systems.

\bibitem[Chuang and Kim, 2023]{ChuKim:OR2023}
Chuang, Y.-T. and Kim, M.~J. (2023).
\newblock Bayesian inventory control: accelerated demand learning via
  exploration boosts.
\newblock {\em Oper. Res.}, 71(5):1515--1529.

\bibitem[Cisternas, 2018]{Cis:RES2018}
Cisternas, G. (2018).
\newblock Two-sided learning and the ratchet principle.
\newblock {\em Rev. Econ. Stud.}, 85(1):307--351.

\bibitem[Erhun et~al., 2008]{ErhKesTay:IIET2008}
Erhun, F., Keskinocak, P., and Tayur, S. (2008).
\newblock Dynamic procurement in a capacitated supply chain facing uncertain
  demand.
\newblock {\em IIE Trans.}, 40(8):733--748.

\bibitem[Feng and Krishnan, 2022]{FenKri:JSR2022}
Feng, S. and Krishnan, T.~V. (2022).
\newblock Contract length determination in the {B2B} service industry: Role of
  economic factors, business relationship, and learning.
\newblock {\em J. Service Res.}, 25(3):422--439.

\bibitem[Freixas et~al., 1985]{FreGuesTir:RES1985}
Freixas, X., Guesnerie, R., and Tirole, J. (1985).
\newblock Planning under incomplete information and the ratchet effect.
\newblock {\em The Review of Economic Studies}, 52(2):173--191.

\bibitem[Fudenberg and Tirole, 1991]{FudTir:MITP1991}
Fudenberg, D. and Tirole, J. (1991).
\newblock {\em Game Theory}.
\newblock MIT Press, Cambridge, MA.

\bibitem[Han, 2016]{Han:PhD2016}
Han, Y. (2016).
\newblock {\em Demand Learning in Two Operations Models}.
\newblock {P}h{D} thesis, Columbia University, New York, NY.

\bibitem[Iskhakov et~al., 2016]{IskRusSch:RES2016}
Iskhakov, F., Rust, J., and Schjerning, B. (2016).
\newblock Recursive lexicographical search: finding all {M}arkov perfect
  equilibria of finite state directional dynamic games.
\newblock {\em Rev. Econ. Stud.}, 83(2):658--703.

\bibitem[Kadiyala et~al., 2020]{KadOzeBen:MS2020}
Kadiyala, B., \"{O}zer, O., and Bensoussan, A. (2020).
\newblock A mechanism design approach to vendor managed inventory.
\newblock {\em Management Sci.}, 66(6):2628--2652.

\bibitem[Kamenica and Gentzkow, 2011]{KamGen:AER2011}
Kamenica, E. and Gentzkow, M. (2011).
\newblock Bayesian persuasion.
\newblock {\em Amer. Econ. Rev.}, 101(6):2590–2615.

\bibitem[Kreps and Wilson, 1982]{KreWil:E1982}
Kreps, D.~M. and Wilson, R. (1982).
\newblock Sequential equilibria.
\newblock {\em Econometrica}, 50(4):863--894.

\bibitem[Laffont and Tirole, 1988]{LafTir:E1988}
Laffont, J.-J. and Tirole, J. (1988).
\newblock The dynamics of incentive contracts.
\newblock {\em Econometrica}, 56(5):1153--1175.

\bibitem[Lariviere and Porteus, 1999]{LarPor:MS1999}
Lariviere, M.~A. and Porteus, E.~L. (1999).
\newblock Stalking information: {B}ayesian inventory management with unobserved
  lost sales.
\newblock {\em Management Sci.}, 45(3):346--363.

\bibitem[Lariviere and Porteus, 2001]{LarPorMSOM2001}
Lariviere, M.~A. and Porteus, E.~L. (2001).
\newblock Selling to the newsvendor: an analysis of price-only contracts.
\newblock {\em Manufacturing Service Oper. Management}, 3(4):293--305.

\bibitem[Leitmann, 1978]{Lei:JOTA1978}
Leitmann, G. (1978).
\newblock On generalized {S}tackelberg strategies.
\newblock {\em J. Optim. Theory Appl.}, 26(4):637--643.

\bibitem[Liu and Rong, 2024]{LiuRon:arXiv2024}
Liu, L. and Rong, Y. (2024).
\newblock No-regret learning for {S}tackelberg equilibrium computation in
  newsvendor pricing games.
\newblock Technical report, arXiv 2404.00203.

\bibitem[Lobel and Xiao, 2017]{LobXia:OR2017}
Lobel, I. and Xiao, W. (2017).
\newblock Technical note---optimal long-term supply contracts with asymmetric
  demand information.
\newblock {\em Oper. Res.}, 65(5):1275--1284.

\bibitem[Mart{\'\i}nez-de Alb{\'e}niz and Simchi-Levi, 2013]{MarSim:POM2013}
Mart{\'\i}nez-de Alb{\'e}niz, V. and Simchi-Levi, D. (2013).
\newblock Supplier-buyer negotiation games: Equilibrium conditions and supply
  chain efficiency.
\newblock {\em Production Oper. Management}, 22(2):397--409.

\bibitem[Maskin and Tirole, 1988]{MasTir:E1988-a}
Maskin, E. and Tirole, J. (1988).
\newblock A theory of dynamic oligopoly. {I}. {O}verview and quantity
  competition with large fixed costs.
\newblock {\em Econometrica}, 56(3):549--569.

\bibitem[Maskin and Tirole, 2001]{MasTir:JET2001}
Maskin, E. and Tirole, J. (2001).
\newblock Markov perfect equilibrium. {I}. {O}bservable actions.
\newblock {\em J. Econom. Theory}, 100(2):191--219.

\bibitem[{McKinsey \& Company}, 2024]{McKinseyGrocery2024}
{McKinsey \& Company} (2024).
\newblock State of grocery {E}urope 2024: Signs of hope.

\bibitem[Mersereau, 2015]{Mer:MSOM2015}
Mersereau, A.~J. (2015).
\newblock Demand estimation from censored observations with inventory record
  inaccuracy.
\newblock {\em Manufacturing Service Oper. Management}, 17(3):335--349.

\bibitem[Mittendorf et~al., 2022]{MitShinYoon:JMR2022}
Mittendorf, B., Shin, J., and Yoon, D.-H. (2022).
\newblock Information disclosure policy and its implications: Ratcheting in
  supply chains.
\newblock {\em Journal of Marketing Research}, 59(2):290--305.

\bibitem[\"Ozer et~al., 2007]{OzeUncWei:EJOR2007}
\"Ozer, O., Uncu, O., and Wei, W. (2007).
\newblock Selling to the ``newsvendor'' with a forecast update: analysis of a
  dual purchase contract.
\newblock {\em European J. Oper. Res.}, 182(3):1150--1176.

\bibitem[Scarf, 1959]{Sca:AMS1959}
Scarf, H. (1959).
\newblock Bayes solutions of the statistical inventory problem.
\newblock {\em Ann. Math. Statist.}, 30:490--508.

\bibitem[Scarf, 1960]{Sca:NRLQ1960}
Scarf, H.~E. (1960).
\newblock Some remarks on {B}ayes solutions to the inventory problem.
\newblock {\em Naval Res. Logist. Quart.}, 7:591--596.

\bibitem[Shen et~al., 2019]{SheChoMin:IJPR2019}
Shen, B., Choi, T.-M., and Minner, S. (2019).
\newblock A review on supply chain contracting with information considerations:
  information updating and information asymmetry.
\newblock {\em Intern. J. Production Res.}, 57(15-16):4898--4936.

\bibitem[Sheng et~al., 2017]{SheGraHuhNag:ORL2017}
Sheng, L., Granot, D., Huh, W.~T., and Nagarajan, M. (2017).
\newblock A dynamic price-only contract: exact and asymptotic results.
\newblock {\em Oper. Res. Lett.}, 45(6):620--624.

\bibitem[Souza, 2023]{talkbusiness2023walmart}
Souza, K. (2023).
\newblock Walmart requiring all suppliers to move to {L}uminate data service.
\newblock TB\&P.

\bibitem[von Stengel and Zamir, 2010]{vonZam:GEB2010}
von Stengel, B. and Zamir, S. (2010).
\newblock Leadership games with convex strategy sets.
\newblock {\em Games Econom. Behav.}, 69(2):446--457.

\bibitem[{Walmart Data Ventures}, 2025]{WalmartETBrowneScintilla2025}
{Walmart Data Ventures} (2025).
\newblock Influenced: Uncovering a viral moment that boosted sales.

\bibitem[{Walmart Inc.}, 2024]{WalmartLuminate2024}
{Walmart Inc.} (2024).
\newblock Walmart announces global launch of {W}almart {L}uminate.

\bibitem[Weitzman, 1980]{Weitzman1980}
Weitzman, M.~L. (1980).
\newblock The ``ratchet principle'' and performance incentives.
\newblock {\em The Bell Journal of Economics}, 11(1):302--308.

\bibitem[Zhang et~al., 2010]{ZhaNagSos:OR2010}
Zhang, H., Nagarajan, M., and So\v{s}i\'c, G. (2010).
\newblock Dynamic supplier contracts under asymmetric inventory information.
\newblock {\em Oper. Res.}, 58(5):1380--1397.

\bibitem[Zhang et~al., 2026]{ZhaLiQinXuZhu:arXiv2026}
Zhang, W., Li, C., Qin, H., Xu, Y., and Zhu, R. (2026).
\newblock Thompson sampling for repeated newsvendor.
\newblock Technical report, arXiv:2502.09900.

\bibitem[Zhao et~al., 2026]{ZhaZhuHas:MS2026}
Zhao, X., Zhu, R., and Haskell, W.~B. (2026).
\newblock Learning to price supply chain contracts against a learning retailer.
\newblock {\em Management Sci.}, 72(3):2168--2187.

\end{thebibliography}

\end{document}